\documentclass[10pt, conference, letterpaper]{IEEEtran}
\IEEEoverridecommandlockouts
\usepackage{amsmath,amssymb,amsfonts}
\usepackage{algorithmic}
\usepackage{graphicx}
\usepackage{textcomp}
\usepackage{xcolor}
\usepackage{cite}
\usepackage{xspace}
\usepackage[ruled,vlined]{algorithm2e}
\usepackage{verbatim} %provide comments enviroment

\newcommand{\vc}[1]{\boldsymbol{#1}}
\newcommand{\nodes}{\ensuremath{\mathcal{V}}\xspace}
\newcommand{\edges}{\ensuremath{\mathcal{E}}\xspace}
\newcommand{\sources}{\ensuremath{\mathcal{S}}\xspace}
\newcommand{\source}{\ensuremath{s}\xspace}
\newcommand{\learners}{\ensuremath{\mathcal{L}}\xspace}
\newcommand{\learner}{\ensuremath{\ell}\xspace}
\newcommand{\x}{\ensuremath{\vc{x}}\xspace}
\newcommand{\paths}{\ensuremath{\mathcal{P}}\xspace}
\newcommand{\p}{\ensuremath{p}\xspace}
\newcommand{\reals}{\mathbb{R}}
\newcommand{\naturals}{\mathbb{N}}

\newcommand{\map}{\mathtt{MAP}}
\newcommand{\cov}{\mathtt{cov}}

\newcommand{\types}{\ensuremath{\mathcal{T}}\xspace}
\newcommand{\type}{\ensuremath{t}\xspace}
\newcommand{\feasibleset}{\mathcal{D}}
\newcommand{\proj}{\mathrm{\Pi}}
\newcommand{\prob}{\mathbf{P}}
\newcommand{\SR}{\ensuremath{\mathtt{TOT}}}
\newcommand{\xset}{\ensuremath{\mathcal{X}}\xspace}

\usepackage{amsthm}
\newtheorem{lemma}{\textbf{Lemma}}
\newtheorem{theorem}{\textbf{Theorem}}
\newtheorem{definition}{\textbf{Definition}}
\newtheorem{ass}{\textbf{Assumption}}
\newtheorem{cor}{\textbf{Corollary}}

\usepackage{multicol}
\usepackage{subcaption}
\usepackage{mathtools}
\usepackage{multirow}
\usepackage{hyperref}

\newcommand{\blu}[1]{{\color{black}#1}}
\newcommand{\CR}[1]{{\color{black}#1}}

\usepackage[font=small,labelfont=bf]{caption}

\newcommand{\arxiv}[2]{#1} % arxiv version

\newenvironment{packed-itemize}
{\leftmargini=15pt\begin{itemize}
\setlength{\itemsep}{-0.3pt}}
{\end{itemize}\vspace{-0.2em}}

\newenvironment{packed-enumerate}
{\leftmargini=15pt\begin{enumerate}
\setlength{\itemsep}{-0.3pt}}
{\end{enumerate}\vspace{-0.2em}}

\def\BibTeX{{\rm B\kern-.05em{\sc i\kern-.025em b}\kern-.08em
    T\kern-.1667em\lower.7ex\hbox{E}\kern-.125emX}}

\begin{document}

\title{Distributed Experimental Design Networks
}

\author{\IEEEauthorblockN{ Yuanyuan Li,\textsuperscript{1} Lili Su,\textsuperscript{1} Carlee Joe-Wong,\textsuperscript{2} Edmund Yeh,\textsuperscript{1} and Stratis Ioannidis,\textsuperscript{1}}
\IEEEauthorblockA{ \textsuperscript{1}Northeastern University, \textsuperscript{2}Carnegie Mellon University, \\
Email: yuanyuanli@ece.neu.edu, l.su@northeastern.edu, cjoewong@andrew.cmu.edu, \{eyeh, ioannidis\}@ece.neu.edu
}
}

\maketitle

\begin{abstract}
As edge computing capabilities increase, model learning deployments in diverse edge environments have emerged. In  experimental design networks, introduced recently, network routing and rate allocation are designed to aid the transfer of data from sensors to heterogeneous learners.  We design efficient experimental design network algorithms that are (a) distributed and (b) use multicast transmissions. This setting poses significant challenges as classic decentralization approaches often operate on (strictly) concave objectives under differentiable constraints. In contrast, the problem we study here has a non-convex, continuous DR-submodular objective, while multicast transmissions naturally result in non-differentiable constraints. From a technical standpoint, we propose a distributed Frank-Wolfe and a distributed projected gradient ascent algorithm that, coupled with a relaxation of non-differentiable constraints, yield allocations within a $1-1/e$ factor from the optimal. 
%\ls{a sentence or two on existing limitations?} We generalize this setting by considering Gaussian sources, incorporating multicasting, but also--crucially--distributed algorithms in this setting. From a technical standpoint, we show that, assuming Gaussian sensor sources still \ls{"still" sounds incremental} yields an continuous DR-submodular experimental design objective. We also propose a distributed Frank-Wolfe algorithms yielding allocations within a $1-1/e$ factor from the optimal. 
Numerical evaluations show that our proposed algorithms outperform competitors with respect to model learning quality.
\end{abstract}

\begin{IEEEkeywords}
Experimental Design, DR-submodularity, Bayesian linear regression, Distributed algorithm.
\end{IEEEkeywords}

\section{Introduction}
%The proliferation of mobile IoT devices and mobile services has resulted in a tremendous amount of data and ever-increasing deployment of machine learning tasks.  
\blu{ We study \emph{experimental design networks}, as introduced by Liu at al.~\cite{liu2022experimental}. In these networks,    illustrated in Fig.~\ref{fig:example},  learners and data sources are dispersed across  different locations in a network. Learners receive streams of data collected from the sources,
and subsequently use them to train models. We are interested in  rate allocation strategies that maximize the quality of model training at the learners, subject to  network constraints.  }\blu{This problem is of practical significance. For instance, in a smart city \cite{mohammadi2018enabling, albino2015smart}, various sensors capture, e.g., image, temperature, humidity, traffic, and seismic measurements, which can help forecast transportation traffic, the spread of disease, pollution levels, the weather, etc. Distinct,  dispersed public service entities, e.g., a transportation authority, an energy company, the fire department, etc., may perform  different training and prediction tasks on these data streams.  %These entities may be geographically dispersed and connected to sensors via the city's network.
}

\blu{ Even though 
the resulting optimization of rate allocations is non-convex, Liu et al.~\cite{liu2022experimental} provide a polynomial-time $(1-1/e)$-approximation algorithm,  exploiting a useful property of the learning objective, namely, continuous DR-submodularity \cite{bian2017guaranteed,soma2015generalization}. %, building upon which a polynomial-time $(1-1/e)$-approximation algorithm was designed~\cite{liu2022experimental}. 
Though~\cite{liu2022experimental} lays a solid foundation for studying this problem,  the algorithm proposed suffers from several limitations. 
First, it is centralized, and requires a full view of network congestion conditions, demand, and learner utilities. This  significantly reduces scalability  when the number of  sources and learners are large.
In addition, it uses unicast transmissions between sources and learners. %, thereby under-utilizing network resources.
In practice, this significantly under-utilizes network resources when learner interests in data streams overlap. 
}

\iffalse 
\blu{Nevertheless, the algorithm proposed by Liu et al.~\cite{liu2022experimental} is centralized, and assumes a full view of network congestion conditions, demand, and learner utilities. Establishing a coordinator for monitoring network conditions counters the advantage of distributed learning across the network. It also poses significant scalability problems, when the number of sensors/data sources, learners, or both, is large. Liu at al.~also model unicast transmissions between sources and learners, thereby under-utilizing network resources. This is important in practice, when desired training datasets may overlap and sending data between them is bandwidth-intensive. This will indeed be the case in, e.g., a smart city scenario, where data such as video captures from multiple locations in the city can collectively amount to terabytes of data collected in a few seconds. Multicasting can significantly improve performance when  learner data needs overlap. }%Finally, Liu et al.~make additional restrictive assumptions, including, e.g., that data samples come from a finite set, which is also not realistic as features are usually continuous values, and labeling noise is homogeneous across sources.

\fi 

%\begin{comment}
\begin{figure}[!t]
\centering
\includegraphics[width=0.9\linewidth]{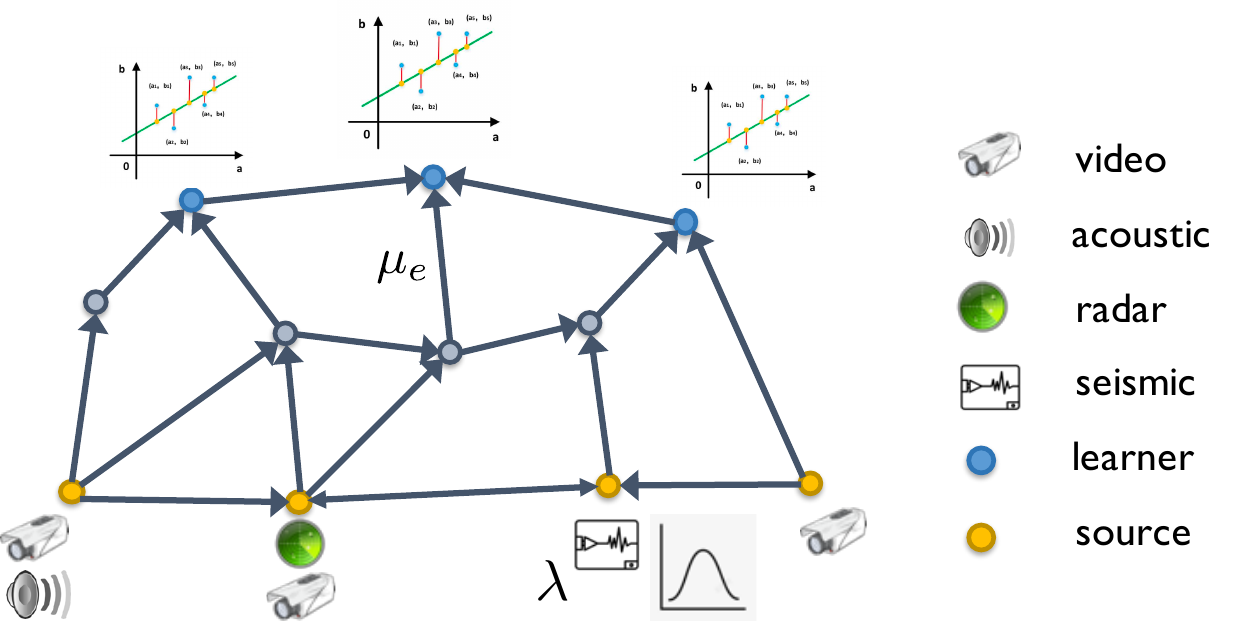}
\caption{An  \emph{experimental design network}~\cite{liu2022experimental}.  Sources (\textcolor{orange}{\texttt{yellow}}) generate streams of data from diverse sensors, e.g., cameras,  microphones, seismic sensors, etc. %, etc., following Poisson process with rate $\lambda$. 
Learners (\textcolor{blue}{\texttt{blue}}) train distinct models over (possibly overlapping) received data. We wish to   
allocate bandwidth to data traffic in a manner that maximizes the social welfare, i.e., the aggregate quality of models  across learners. %Liu et al. \cite{liu2022experimental} consider generating feature from a finite set, while we consider Gaussian sources.
}
\label{fig:example}
\end{figure}
\blu{ In this paper, we aim to design efficient algorithms that are (a) distributed and (b) use multicast transmission. Achieving this goal is far from trivial. 
First, classic decentralization approaches, such as, primal-dual algorithms \cite{low1999optimization,srikant2004mathematics,lun2005achieving}, often operate on (strictly) concave objectives. In contrast, the problem we study here has a non-concave, continuous DR-submodular objective. Second, multicast transmissions naturally result in non-differentiable constraints (see \cite{srikant2004mathematics,lun2006minimum}, but also Eqs.~\eqref{cons: link capacity} and \eqref{cons: source} in Sec.~\ref{sec: problem formulation}). This further hinders  standard decentralization techniques. Our {\bf contributions} are  as follows:}% 
\begin{packed-itemize}
    \item \blu{We incorporate multicast transmissions to experimental design networks. This is more realistic when learning jointly from common sensors, and yields a significantly higher throughput in comparison to unicast transmissions.}
    \item We prove that, assuming Poisson data streams in steady state and  Bayesian linear regression  as a learning task, as in \cite{liu2022experimental}, our experimental design objective remains continuous DR-submodular.
    \item  We construct \emph{both} centralized \emph{and} distributed algorithms within a $1-1/e$ factor from the optimal in this setting. For the latter, we make use of a primal dual technique that addresses both the non-differentiability of constituent multicast constraints, as well as the lack of strict convexity exhibited by our problem.
    \item We conduct extensive simulations over both synthetic and backbone network topologies. Our proposed algorithms outperform all competitors w.r.t. the quality of model estimation, and our distributed algorithms perform closely to their centralized versions.
\end{packed-itemize} 

\blu{From a technical standpoint, we couple the Frank-Wolfe algorithm by Bian et al.~\cite{bian2017guaranteed}, also used by Liu et al.~\cite{liu2022experimental}, with a nested distributed step; the latter  deals with the non-differentiability of multicast constraints through an $l_p$ relaxation of the max norm.
We also implement two additional extensions \CR{sketched out} by Liu et al.~\cite{liu2022experimental}: we consider (a) Gaussian data sources, that are (b) subject to  heterogeneous noise. We incorporate both in our mathematical formulation and theoretically and experimentally characterize performance under these extensions.   Gaussianity  %generalize the classic definition of experimental design \cite{boyd2004convex}. Furthermore, it 
requires revisiting how gradient estimation is performed, compared to Liu et al., as well as devising new estimation bounds.} %We overcome both challenges, as well as challenges related to the non-differentiability of multicast constraints, when constructing our distributed optimization algorithm. 

The remainder of this paper is organized as follows. Sec.~\ref{sec: related work} provides a literature review. We introduce our distributed model in Sec.~\ref{sec: problem formulation}. Sec.~\ref{sec: centralized algorithm} describes our analysis of the problem and proposed centralized algorithm, while Sec.~\ref{sec: distributed algorithm} describes our distributed algorithm. We  propose additional distributed algorithms in Sec.~\ref{sec: extensions}. We present numerical experiments in Sec.~\ref{sec: numerical evaluation}, and conclude in Sec.~\ref{sec: conclusion}.

\section{Related Work}
\label{sec: related work}
\noindent\textbf{Experimental Design.} \CR{As discussed by Liu et al.~\cite{liu2022experimental}, experimental design is classic under a single user with an experiment budget constraint~\cite{boyd2004convex, pukelsheim2006optimal}, while the so-called D-Optimality criterion is a popular objective~\cite{horel2014budget,guo2019accelerated,gast2020linear, guo2018experimental,huan2013simulation}.} \blu{Liu et al. \cite{liu2022experimental} are the first to extend this objective to the context of experimental design networks. %and to propose a poly-time algorithm with a provable approximation guarantee. 
As discussed in the introduction, we deviate from Liu et al.~by proposing a decentralized algorithm and considering multicast transmissions; both are practically important and come with technical challenges. }
\blu{Liu et al.~make additional restrictive assumptions, including, e.g., that data samples come from a finite set and that labeling noise is homogeneous across sources, but mention that their analysis could be extended to amend these assumptions. We implement this extension by considering Gaussian sources and noise heteroskedasticity, and proving gradient estimation bounds using appropriate Chernoff inequalities (see, e.g., Lem.~\ref{lem: HEAD bound4}). }

\noindent\textbf{Submodular Maximization.} Submodularity is traditionally \CR{explored within} the context of set functions \cite{calinescu2011maximizing}, but can also be extended to functions over the integer lattice \cite{soma2015generalization} and the continuous domain~\cite{bian2017guaranteed}. Maximizing a monotone submodular function subject to a matroid constraint is classic. Krause and Golovin \cite{krause2014submodular} show that the greedy algorithm achieves a $1/2$ approximation ratio. Calinescu et al. \cite{calinescu2011maximizing} propose a \emph{continuous greedy} algorithm improving the ratio to $1-1/e$ that applies a Frank-Wolfe (FW)~\cite{bertsekas1999nonlinear}  variant to the multilinear extension of the submodular objective. With the help of auxiliary potential functions, %Filmus and Ward~\cite{filmus2014monotone} run a non-oblivious local search after the greedy algorithm, and also produce a $1-1/e$ approximation ratio. %Further improvements are made by Sviridenko et al. \cite{sviridenko2017optimal} for a more restricted class of submodular functions with bounded curvature. 
Bian et al.~\cite{bian2017guaranteed} show that the same FW variant can be used to maximize continuous DR-submodular functions within a $1-1/e$ ratio. The centralized algorithm by Liu et al.~\cite{liu2022experimental}, and ours, are applications of the FW variant~\cite{bian2017guaranteed}; in both cases, recovering their guarantees requires devising novel gradient estimators and bounding their estimation accuracy. We also depart by considering a distributed version of this algorithm, where each node accesses  only neighborhood knowledge. %This requires distributing constituent steps of FW via a primal-dual approach; as discussed below, the latter is a priori designed for convex optimization problems, a property that does not hold in our setting.

\noindent\textbf{Convergence of Primal-Dual Algorithms.} Nedi{\'c} and Ozdaglar~\cite{nedic2009subgradient} propose a subgradient algorithm for generating approximate saddle-point solutions for a convex-concave function. Assuming  Slater's condition and bounded Lagrangian gradients, they provide bounds on the primal objective function. 
Alghunaim and Sayed \cite{alghunaim2020linear} prove linear convergence for primal-dual gradient methods. The methods apply to augmented Lagrangian formulations, whose primal objective is smooth and strongly convex under equality constraints (ours are inequality constraints). 
Lyapunov equations are usually employed for a continuous version of the primal-dual gradient algorithm \cite{srikant2004mathematics}; this requires objectives to be strictly concave. Feijer and Paganini \cite{feijer2010stability} prove the stability of primal–dual gradient dynamics with concave objectives through Krasovskii’s method and the LaSalle invariance principle. We follow  \cite{feijer2010stability}  to ensure  convergence of our decentralized algorithm. %\blu{and combine this approach with the FW variant algorithm of Bian et al.~\cite{bian2017guaranteed}, yielding a $1-1/e$ approximation guarantee.}

\noindent\textbf{Distributed Algorithms.} Distributed algorithms for the maximization of strictly concave objectives under separable constraints are classic (see, e.g., \cite{low1999optimization,srikant2004mathematics,lun2006minimum,bertsekas2015parallel}). %Low and Lapsley \cite{low1999optimization} \arxiv{consider a unicast multiple-source/single-sink network, and }{}maximize a strictly concave utility w.r.t. rate allocation under bandwidth constraints. They solve the dual problem using a gradient projection (dual) algorithm, and prove optimality guarantees. 
%Another approach for solving network resource allocation problems exactly is the classic primal-dual algorithm (see, e.g.,~\cite{srikant2004mathematics}): intuitively, a primal algorithm is implemented at the sources and a dual algorithm at the links. Lun et al.~\cite{lun2006minimum} \arxiv{consider a single-source/multiple-sinks setting, combined with hop-to-hop multicast. They again maximize concave utilities w.r.t.~rate allocation under bandwidth constraints. They }decentralize problems with \arxiv{both linear and }strictly concave utilities through\arxiv{: subgradient optimization with $\epsilon$-relaxation method~\cite{bertsekas2015parallel} and the primal-dual method mentioned before, respectively}{ it}. 
%
%All the above frameworks critically rely on strict convexity of the problem objectives, and they only consider homogeneous data. On the contrary, we consider several types of data, which share the network bandwidth together, and 
%All of the above frameworks pertain to convex optimization problems.
Our objective is continuous DR-submodular; thus, these methods do not directly extend to our setting. 
Tychogiorgos et al.~\cite{tychogiorgos2013non} provide the theoretical foundations for distributed dual algorithm solution of non-convex problems.
%duality gap is zero, when primal variable is continuous around at least one of the optimal Lagrange multiplier vectors and Slater’s condition holds. However, this distributed dual algorithm computes primal variables by maximizing Lagrangian given current dual variables, which is non-convex optimization in our setting. 
\blu{Mokhtari et al.~\cite{mokhtari2018decentralized} propose a partially decentralized continuous greedy algorithm for DR-submodular maximization subject to down-closed convex constraints and prove a $1-1/e$ guarantee. However, the conditional gradient update step they propose requires global information and remains centralized.}
Our analysis requires combining above techniques, 
%particularly the Feijer and Paganini~\cite{feijer2010stability} primal-dual algorithm described above, 
with the (centralized) Frank-Wolfe variant by Bian et al.~\cite{bian2017guaranteed}, which yields a $1-1/e$ approximation guarantee. Doing so requires dealing with both the lack of strict convexity of the constituent problem, as well as the non-differentiability of multicast constraints.

\section{Problem Formulation}
\label{sec: problem formulation}

\CR{Our model, and its exposition below, follows closely Liu et al.~\cite{liu2022experimental}: for consistency, we use the same notation and terminology. We depart in multiple ways. First and foremost, we (a) seek a \emph{distributed algorithm} determining the rate allocation, requiring knowledge only from itself and its neighbourhoods, while the centralized algorithm requires global information, and (b) we extend the analysis from unicast to multicast, allowing sharing of the same traffic across learners, thereby increasing throughput. We also implement two potential extensions drafted by Liu et al.~\cite{liu2022experimental}: we (c) generate features from Gaussian sources instead of a finite set, changing subscripts of variables from feature $\x$ to source $\source$, (d) adopt source instead of hop-by-hop routing, which changes the optimized random variables, leading to a general directed graph instead of a DAG (directed acyclic graph), and (f) incorporate heterogeneous noise over source $\source$ and types $\type$ instead of homogeneous constant noise. }

%\subsection{Experimental Design Network Model}
%\label{sec: network model}
\noindent\textbf{Network.} We model 
the system %is modelled
as a  multi-hop network with a topology represented by a  \emph{directed graph} $\mathcal{G}(\nodes, \edges)$, where $\nodes$ is the set of nodes and $\edges\subset \nodes\times\nodes$ is the set of links.  Each link $e=(u,v)\in\edges$ has a link capacity $\mu^e\geq 0$. Sources $\sources\subset \nodes $ generate data streams, while learners $\learners\subset \nodes$ reside at sinks. 

\noindent\textbf{Data Sources.} Each data source $\source\in\sources$ generates a sequence of labeled pairs $(\x,y)\in \reals^d\times \reals$ of type $\type$ according to a Poisson process of rate $\lambda_{\source,\type}\geq 0,$ corresponding to measurements or experiments the source conducts (each pair is a new measurement). %Once an experiment with features $\x$ is conducted, the source labels it with $y \in \reals$ of type $\type$ out of a set of possible types $\types$. 
Intuitively, features $\x$ correspond to covariates  in an experiment (e.g., pixel values in an image, etc.),  label types $\type\in\types$ correspond to possible measurements (e.g., temperature, radiation level, etc.), and labels $y$ correspond to the actual measurement value collected (e.g., $23^{\circ}\text{C})$.

%\ls{We do not need to repeatedly mentioned liu et al. I made some minor changes in the following paragraph.}
  \blu{
The generated data follows a linear regression model~\cite{james2013introduction,gallager2013stochastic},} i.e., for every type $\type\in\types$ from source $\source \in \sources$, there exists a $\vc{\beta}_\type\in\reals^d$ such that $y = \x^\top\vc{\beta}_\type + \epsilon_{\source, \type}$ where $\epsilon_{\source, \type} \in \reals$ are i.i.d.~zero mean normal noise variables with variance $\sigma^2_{\source, \type} > 0$.
  Departing from %the standard settings \ls{(???) check this claim}  
  Liu et al.~\cite{liu2022experimental}, but also from classic experimental design \cite{boyd2004convex}, where $\source$ samples feature vectors $\x\in\reals^d$ from a finite set, we assume that they are sampled from a \emph{Gaussian} distribution $N(\vc{0}, \vc{\Sigma}_{\source})$.
Again in contrast to \cite{liu2022experimental}, we allow for \emph{heterogeneous} (also known as~heteroskedastic) noise levels $\sigma_{\source,\type}$ across \emph{both} sources and experiment types.

%   \blu{
%   Departing from the standard settings (???) where the features are supported on finite set \cite{liu2022experimental,boyd2004convex}, 
% %Contrary to , who consider features sampled from a  finite set, 
% we assume that source $\source$  feature $\x$ is sampled from a \emph{Gaussian} distribution $N(\vc{0}, \vc{\Sigma}_{\source})$, over $\reals^d$.
% Moreover, as in \cite{liu2022experimental}, we assume that the generated data follows a linear regression model~\cite{james2013introduction,gallager2013stochastic}; } that is, for every type $\type\in\types$ from source $\source \in \sources$, there exists a $\vc{\beta}_\type\in\reals^d$ s.t. $y = \x^\top\vc{\beta}_\type + \epsilon_{\source, \type}$ where $\epsilon_{\source, \type} \in \reals$ are i.i.d.~zero mean normal noise variables with variance $\sigma^2_{\source, \type} > 0$. \blu{In contrast to \cite{liu2022experimental}, the noise level is \emph{heterogeneous} (a.k.a.~heteroskedastic) across both sources and experiment types.}

\noindent\textbf{Learners and Bayesian Linear Regression}. 
\blu{%As in \cite{liu2022experimental}, 
Each learner $\learner \in \learners$ wishes to learn a model $\vc{\beta}_{\type^\learner}$ for some type $\type^\learner\in \types$, via \emph{Bayesian linear regression}. 
In particular, each
$\learner$ has a Gaussian prior $N(\vc{\beta}_0^\learner,\vc{\Sigma}_0^\learner)$ on the model $\vc{\beta}_{\type^\learner}$ it wishes to estimate. %The learner $\learner$ wishes to consume received data pairs $(\x,y)$ of type $\type^\learner \in \types$, and subsequently estimate $\vc{\beta}_{\type^\learner}$
%through }.
%\noindent\textbf{Learner MAP Estimation.}
%Different from the classic experimental design problem~\cite{boyd2004convex}, as well as Liu et al.~\cite{liu2022experimental}, where the set of features is assumed to be finite, we consider a uncountable feature set. 
We assume that the system operates for a data acquisition time period $T$. % at the end of which each learner $\learner\in\learners$ estimates $\vc{\beta}_{\type^\learner}$ based on the data it has received during this period via MAP estimation.
%The learner's goal is to estimate the model parameter $\vc{\beta}$ from samples $\{(\vc{x_i},y_i)\}_{i=1}^n$. In Bayesian linear regression, it is additionally assumed that $\vc{\beta}$ is sampled from a prior normal distribution with  mean $\vc{\beta}_0\in\reals^d$ and covariance $\vc{\Sigma}_0\in\reals^{d\times d}$ (i.e., $\vc{\beta} \sim N(\vc{\beta}_0, \vc{\Sigma}_0)$).
% by Eq.~\eqref{eq:map}. 
%In particular, 
Let $n_{\source}^\learner\in\naturals$ be the cumulative number of pairs $(\x,y)$ from source $\source$ collected by learner $\learner$ during this period, and $\vc{n}^\learner = [n^{\learner}_{\source}]_{\source\in \sources}$ the vector of arrivals at learner $\learner$ from all sources.
We denote by $\x_{\source,i}^\learner$, $y_{\source,i}^\learner$ the $i$-th feature and label generated from source $\source$ to reach learner $\learner$, and $\vc{X}^{\learner} = [[ \x_{\source,i}^{\learner} ]_{i = 1}^{n_{\source}^\learner}]_{\source \in \sources}$, $\vc{y}^\learner = [[ y_{\source,i}^{\learner} ]_{i = 1}^{n_{\source}^\learner}]_{\source \in \sources}$ the feature matrix and label vector, respectively, received at learner $\learner$ up to time $T$. Then, maximum a-posterior (MAP) estimation \cite{gallager2013stochastic,liu2022experimental} at learner $\learner$ amounts to:
%\begin{align}
%\label{eq:map}
%   \hat{\vc{\beta}}^\map_{\learner} = \big( \sum_{\source \in \sources}\sum_{i=1}^{n_\source^\learner} \frac{\x_{\source,i}^\learner (\x_{\source,i}^\learner)^{\top}}{\sigma_{\source, \type^\learner}^2} + \vc{\Sigma}_\learner^{-1}\big)^{-1} \big( \sum_{\source \in \sources}\sum_{i=1}^{n_\source^\learner} \frac{\x_{\source,i}^\learner (\x_{\source,i}^\learner)^{\top}}{\sigma_{\source, \type^\learner}^2} + \vc{\Sigma}_\learner^{-1}\vc{\beta}_{\learner}\big), 
%\end{align}
%\cite{gallager2013stochastic}:
\begin{align}
\label{eq:map}
\begin{split}
   \hat{\vc{\beta}}_\map^\learner = ((\vc{X}^\learner)^\top (\vc{\tilde{\Sigma}}^\learner)^{-1} \vc{X}^\learner + (\vc{\Sigma}_0^\learner)^{-1})^{-1}\cdot\\ ((\vc{X}^\learner)^{\top} (\vc{\tilde{\Sigma}}^\learner)^{-1} \vc{y}^\learner + (\vc{\Sigma}_0^\learner)^{-1}\vc{\beta}_0^\learner ),
   \end{split}
\end{align}
where $\vc{\tilde{\Sigma}}^\learner = \reals^{n^\learner\times n^\learner}$ is a diagonal noise covariance matrix, containing corresponding source noise covariances $\sigma_{s,t^\learner}^2$ in its diagonal. }
\blu{
The quality of this estimator is determined by the error covariance \cite{gallager2013stochastic}, i.e., the $d\times d$ matrix:%of the estimation error difference $\hat{\vc{\beta}}_\map-\vc{\beta}$ \cite{gallager2013stochastic}:
\begin{equation}
 \label{eq:cov}
\cov(\hat{\vc{\beta}}_\map^\learner-\vc{\beta}_{t^\learner}) = \big((\vc{X}^\learner)^{\top} (\vc{\tilde{\Sigma}}^\learner)^{-1} \vc{X}^\learner  + (\vc{\Sigma}_0^\learner)^{-1}\big)^{-1}.  
\end{equation}
The covariance summarizes the estimator quality in all directions in $\reals^d$: directions $\vc{x}\in \reals^d$ of high variance (e.g., eigenvectors in which eigenvalues of  $\cov(\hat{\vc{\beta}}_\map^\learner-\vc{\beta}_{t_\learner})$  are high) are directions in which the prediction  $\hat{y}=\vc{x}^\top \hat{\vc{\beta}}_\map^\learner$ will have the highest prediction error. % (see also~\cite{liu2022experimental}). 
}

\noindent\textbf{Network Constraints}. Data pairs $(\x,y)\in\reals^d\times\reals$ of type $\type \in \types$ generated by sources are transmitted over paths in the network and eventually delivered to learners. \blu{Furthermore, we consider the \emph{multirate multicast} transmission~\cite{srikant2004mathematics}, which saves network resources compared to unicast. %Our network design seeks a \emph{distributed} allocation of  data for learning tasks under network capacity constraints. 
That is, we assume  that each source $\source$  has a set of paths $\paths_{\source, \type}$ over which data pairs of type $\type$ are routed:
each path $\p \in \paths_{\source, \type}$ links to a different learner.
%, i.e., $\p, \p' \in \paths_{\source, \type}$ and $\p[-1] \neq \p'[-1]$.
Note that $|\paths_{\source, \type}|$, the size of $\paths_{\source, \type}$, equals  the number of learners with type $t$.   We denote the \emph{virtual} rate~\cite{srikant2004mathematics}, 
with which data pairs of type $\type$ from source $\source$ are transmitted through path $\p\in\paths_{\source, \type}$, as $\lambda_{\source, \type}^{\p} \geq 0.$ } Let 
$
P_\SR =\sum_{\source \in \sources, \type \in \types}|\paths_{\source, \type}|
$
be the total number of paths. We refer to the vector
\begin{equation}
\label{eq: lambda}
    \vc{\lambda} =  [\lambda_{\source, \type}^{\p}]_{\source\in\sources, \type \in \types, \p \in \paths_{\source, \type}} \in \reals_+, ^{P_\SR}
\end{equation}
as the global rates allocation. To satisfy multicast link capacity constraints, for each link $e\in\edges$, we must have
\begin{equation}
\label{cons: link capacity}
    \sum_{\source \in \sources, \type \in \types} \max_{\p \in \paths_{\source, \type}: e \in \p} \lambda_{\source, \type}^{ \p} \leq \mu^e.
\end{equation}
Note that only data pairs of the same type generated and same source can be multicast together.
For learner $l$, we denote by 
\begin{equation}
\label{cons: learner}
    \lambda_{\source}^\learner = \lambda_{\source, \type^{\learner}}^{\p}
\end{equation}
the incoming traffic rate of type $\type^{\learner}$ at $\learner\in\learners$ from source $\source\in\sources$. Note that $\p \in \paths_{\source,\type^{\learner}}$ and $\learner$ is the last node of $\p$.
At source $\source$, for each $\type\in\types$, we have the constraints:
\begin{equation}
\label{cons: source}
 \max_{\p \in \paths_{\source, \type}} \lambda_{\source, \type}^{\p} \le \lambda_{\source,\type}.
\end{equation}
\blu{Note that the left hand sides of both constraints in \eqref{cons: link capacity} and \eqref{cons: source} are non-differentiable.} 
\CR{We adopt the following assumption on the network substrate (Asm.~1 in \cite{liu2022experimental})}:
\begin{ass}
\label{asm: poisson}
For $\vc{\lambda}\in\feasibleset$, the system is stable and, in steady state, pairs $(\x,y)\in\reals^{d}\times \reals$ of type $t^\learner$ arrive at learner $\learner\in \learners$ according to $|\sources|$ independent Poisson processes with rate $\lambda^\learner_{\source}$.
\end{ass}
\CR{As discussed in \cite{liu2022experimental}}, this is satisfied if, e.g., the network is a Kelly network~\cite{kelly2011reversibility} where Burke's theorem holds~\cite{gallager2013stochastic}.

\noindent\textbf{D-Optimal Design Objective.} The so-called D-optimal design objective~\cite{boyd2004convex} for learner $\learner$ is given by:
\blu{\begin{equation}
%\begin{split}
\label{eq: Gaussian D-optimal}
   G^{\learner}(\vc{X}^\learner, \vc{n}^\learner%; \vc{\sigma}_{\type^\learner}, \vc{\Sigma}_\learner
   ) = \log\det (\cov(\hat{\vc{\beta}}_\map^{\learner}-\vc{\beta}_{\type^\learner})), %\\
   %= \log\det\big( \sum_{\source \in \sources}\sum_{i=1}^{n_\source^\learner} \textstyle \frac{\x_{\source,i}^\learner (\x_{\source,i}^\learner)^{\top}}{\sigma_{\source, \type^\learner}^2} + \vc{\Sigma}_\learner^{-1}\big).
%\end{split}
\end{equation}}%
\blu{where the covariance is given by  Eq.~\eqref{eq:cov}. As the latter summarises the expected prediction error in all directions, minimizing the $\log \det$ (i.e., the sum of logs of eigenvalues of the covariance) imposes an overall bound on  this error. 
}

\noindent\textbf{Aggregate Expected Utility Optimization.}
Under Asm. \ref{asm: poisson}, the arrivals of pertinent data pairs at learner $\learner$ from source $\source$ form a Poisson process with rate $\lambda_{\source}^\learner$. The 
PMF of arrivals is:
\begin{align}
\label{eq:Poisson prod}
    \prob[\vc{n}^\learner=\vc{n}] = \prod_{\source\in\sources} \frac{(\lambda_{\source}^\learner T)^{n_{\source}^\learner}  e^{-\lambda_{\source}^\learner T}}{n_{\source}^\learner!}, 
\end{align}
for all $\vc{n}=[n_{\source}]_{\source\in \sources}\in\naturals^{|\sources|}$ and $\learner\in\learners$. Then, the PDF (probability distribution function) of features is:
\begin{equation}
\label{eq:Gaussian prod}
    f(\vc{X}^{\learner} = \vc{X}) = \prod_{\source\in\sources} \prod_{i=1}^{n_{\source}^{\learner}} \frac{1}{(2\pi)^{\frac{d}{2}} \sqrt{|\vc{\Sigma}_{\source}|}} e^{-\frac{1}{2} \x_{\source,i}^\top \vc{\Sigma}_{\source}^{-1} \x_{\source,i}},
\end{equation}
for all $\vc{X} = [[ \x_{\source,i} ]_{i = 1}^{n_{\source}}]_{\source \in \sources}\in\reals^{\sum_{\source} n_{\source}}$ and $\learner\in\learners$. We define the utility at learner $\learner\in\learners$ as its  expected D-optimal design objective, namely:
\begin{align*}
\begin{split}
    & U^\learner (\vc{\lambda}^\learner) = \mathbb{E}_{\vc{n}^\learner}\big[\mathbb{E}_{\vc{X}^\learner} [G^\learner(\vc{X}^\learner,\vc{n}^\learner)|\vc{n}^\learner]\big] \\
    & = \sum_{\vc{n}\in\naturals^{|\sources|}}  \prob[\vc{n}^\learner=\vc{n}] \int_{\vc{X}\in\reals^{\sum_{\source} n_{\source}}} G^{\learner}(\vc{X}, \vc{n}) f(\vc{X}^{\learner} = \vc{X}) \mathrm{d} \vc{X},
\end{split}
\end{align*}
where $\vc{\lambda}^\learner = [\lambda_{\source}^\learner]_{\source \in \sources}$, \blu{and D-optimal design objective $ G^{\learner}$ is given by Eq.~\eqref{eq: Gaussian D-optimal}}. We wish to solve the following problem:
\begin{subequations}
\label{prob: utility}
\begin{align}
    \text{Maximize:} \quad &U(\vc{\lambda}) 
    = \sum_{\learner\in\learners}(U^\learner(\vc{\lambda}^\learner) - U^\learner(\mathbf{0})),\label{eq: objective function} \\
    \text{s.t.} \quad &\vc{\lambda} \in \feasibleset,
\end{align}
\end{subequations}
where $U^\learner(\mathbf{0})$ is a lower bound\footnote{This is added to ensure the non-negativity of the objective, which is needed to state guarantees in terms of an approximation ratio (c.f.~Thm.~\ref{thm: FW}) \cite{liu2022experimental}.} for $U^\learner(\vc{\lambda}^\learner)$ and feasible set $\feasibleset$ is defined by constraints \eqref{eq: lambda}-\eqref{cons: source}. Note that the feasible set is a down-closed convex set \cite{bian2017guaranteed}. However, this problem is not convex, as the objective is non-concave.

\section{Centralized Algorithm}
\label{sec: centralized algorithm}
In this section, we propose a centralized polynomial-time algorithm with a new gradient estimation, as required by the presence of Gaussian sources, to solve Prob. \eqref{prob: utility}. By establishing submodularity, we achieve an optimality guarantee of $1-1/e$.

\subsection{DR-submodularity}
To solve this non-convex problem, we utilize a key property here: \emph{diminishing-returns submodularity}, defined as follows:
\begin{definition} 
[DR-Submodularity \cite{bian2017guaranteed,soma2015generalization}]
\label{def: submodular}
A function $f: \naturals^p \to \reals$ is called \emph{diminishing-returns (DR) submodular} iff for all $\x, \vc{y} \in \naturals^p$ such that $\x\leq \vc{y}$ and all $k\in \naturals$,
\begin{equation}
\label{eq:drsub}
    f(\x + k \vc{e}_j) - f(\x) \geq f(\vc{y} + k\vc{e}_j) -f(\vc{y}),
\end{equation}
for all $j=1,\ldots,p$, where $\vc{e}_j$ is the $j$-th standard basis vector. 
Moreover, if Eq.~\eqref{eq:drsub} holds for a real valued function $f: \reals^p_+ \to \reals$ for all  $\x, \vc{y} \in \reals^p$ s.t. $\x\leq \vc{y}$ and all $k\in \reals_+$, the function is called \emph{continuous DR-submodular}. 
\end{definition}
The following theorem establishes that objective \eqref{eq: objective function} is a continuous DR-submodular function. 
\begin{theorem}
\label{thm: submodular}
Objective $U (\vc{\lambda})$ is (a) monotone-increasing and (b) continuous DR-submodular with respect to $\vc{\lambda}$. Moreover, the partial derivative of $U$ is:
\begin{equation}
\label{eq: derivative}
    \frac{\partial U}{\partial \lambda_{\source, \type}^{\p}} = \sum_{n=0}^{\infty} \Delta_{\source}^{\learner}(\vc{\lambda}^\learner,n) \cdot \prob[n_{\source}^{\learner}= n] \cdot T,
\end{equation}
where $\type = \type^{\learner}$, $\learner$ is the last node of $\p$, the distribution $\mathrm{P}$ is Poisson described by \eqref{eq:Poisson prod}, with parameters governed by $\lambda_{\source}^{\learner} T$, and
\begin{equation}
\label{eq: delta}
\begin{split}
    \Delta_{\source}^{\learner}(\vc{\lambda}^\learner,n) = \mathbb{E}_{\vc{n}^{\learner}}\left[\mathbb{E}_{\vc{X}^{\learner}} [G^{\learner}(\vc{X}^{\learner},\vc{n}^{\learner})|\vc{n}^{\learner}]|n_{\source}^{\learner}= n+1\right] \\
    - \mathbb{E}_{\vc{n}^{\learner}}\left[\mathbb{E}_{\vc{X}^{\learner}} [G(\vc{X}^{\learner},\vc{n}^{\learner})|\vc{n}^{\learner}]|n_{\source}^{\learner}= n\right].
\end{split}
\end{equation}
\end{theorem}
\begin{comment}
\noindent\textbf{Proof Sketch.} We first prove function $G^{\learner}$, defined in Eq.~\eqref{eq: Gaussian D-optimal}, is monotone-increasing and DR-submodular w.r.t. $\vc{n}^\learner$ following Sylvester's determinant identity \cite{akritas1996various, horel2014budget}. Then, through proving that the expectation of infinity DR-submodularity preserves DR-submodularity, we prove the Hessian of $U^\learner$ being non-positive following Thm.~1 in \cite{liu2022experimental}. $\hfill\qed$
\end{comment}

The proof is in \arxiv{App.~\ref{app: proof_submodular}}{\CR{our technical report \cite{arxviv}}}. Obj.~\eqref{eq: objective function} contains two layers of expectations and a different D-optimal design objective from~\cite{liu2022experimental}. This is a consequence of the Gaussianity and heterogeneity of sources. In turn, this also requires a different argument in establishing  the continuous DR-submodularity.

\subsection{Algorithm Overview}
We follow the Frank-Wolfe variant for monotone continuous DR-submodular function maximization by Bian et al. \cite{bian2017guaranteed} and Liu et al. \cite{liu2022experimental}, but deviate in estimating the gradients of objective $U$. \CR{The proposed algorithm is summarized in Alg.~\ref{alg: FW}.}

\begin{algorithm}[t]
% \small
\caption{Frank-Wolfe Variant}
\label{alg: FW}
\LinesNumbered
\KwIn{ $U:\feasibleset\to\reals_+$,  $\feasibleset$, stepsize $\delta\in(0,1]$.}
$\vc{\lambda}^0=0, \eta=0, k=0$ \\
\While{$\eta<1$}{
    find direction $\vc{v}(k)$, s.t. $  \vc{v}(k) = \mathop{\arg\max}_{\vc{v}\in \mathcal{D }}\langle \vc{v},\widehat{\nabla U(\vc{\lambda}(k))}\rangle$ \\
    $\gamma_k = \min \{ \delta, 1-\eta \}$ \\
    $\vc{\lambda}(k+1) = \vc{\lambda}(k) + \gamma_k \vc{v}(k)$, $\eta=\eta+\gamma_k$, $k=k+1$
}
\Return $\vc{\lambda}(K)$
\end{algorithm}

\noindent{\textbf{Frank-Wolfe Variant.}} Starting from $\vc{\lambda}(0) = \vc{0}$, FW iterates:
\begin{subequations}
\label{FW}
\begin{align}
    & \vc{v}(k) = \mathop{\arg\max}_{\vc{v}\in \mathcal{D }}\langle \vc{v},\widehat{\nabla U(\vc{\lambda}(k))} \rangle, \label{FW_1} \\
    & \vc{\lambda}(k+1) = \vc{\lambda}(k) + \gamma \vc{v}(k), \label{FW_2}
\end{align}
\end{subequations}
where $\widehat{\nabla U(\cdot)}$ is an estimator of the gradient $\nabla U$, and $\gamma$ is an appropriate stepsize. We will further discuss how to estimate the gradient in Sec.~\ref{sec: gradient estimation}.
This algorithm achieves a $1-\frac{1}{e}$ approximation guarantee, characterized by:

\begin{theorem}
\label{thm: FW}
Let $\lambda_\mathrm{MAX} = \max_{\vc{\lambda}\in\mathcal{D}}\|\vc{\lambda}\|_1$. Then, for any $0<\epsilon_0,\epsilon_1<1$, there exists $K = O(\frac{\epsilon_0}{P_\SR  (|\sources|-1)} \epsilon_1)$, $n'=O( \lambda_\mathrm{MAX} T + \ln \frac{1}{\epsilon_1})$, $N_1=N_2=\Omega(\sqrt{\ln \frac{P_\SR K}{\epsilon_0} \cdot (n'+1)}TK)$, s.t., the FW variant algorithm terminates in $K$ iterations, and uses $n'$ terms in the sum, $N_1$ samples for $\vc{n}$, and $N_2$ samples for $\vc{X}$ in estimator \eqref{eq: estimator}. Thus, with probability greater than $1-\epsilon_0$, the output solution $\vc{\lambda}(K) \in \mathcal{D}$ \CR{of Alg.~\ref{alg: FW} }satisfies:
\begin{equation}\label{eq: final guarantee}
    U(\vc{\lambda}(K)) \ge (1 - e^{\epsilon_1-1})  \max_{\vc{\lambda}\in\mathcal{D}}U(\vc{\lambda}) - \epsilon_2,
\end{equation}
where $\epsilon_2$ determined by $\epsilon_0,\epsilon_1$ and network parameters: $\epsilon_2= (T^2 P_\SR\lambda_{\mathrm{MAX}}^2 + 2 \lambda_{\mathrm{MAX}}) \frac{1}{K} \max_{\learner \in \learners, \source \in \sources} \log\left(1 + \frac{\lambda_{\mathrm{MAX}}(\vc{\Sigma}_0^\learner)c_{\source}^2}{\sigma_{\source,\type}^2} \right)>0$, $c_\source = 4 \sqrt{\lambda_{\mathrm{MAX}}(\vc{\Sigma}_\source)} \sqrt{d} + 2 \sqrt{\lambda_{\mathrm{MAX}}(\vc{\Sigma}_\source)} \sqrt{\log \frac{1}{\delta}}$, and $\delta = O\left(\frac{\epsilon_0}{P_\SR K |\sources|n'}\right)$.
\end{theorem}
\begin{comment}
\noindent\textbf{Proof Sketch.} Combining Lem.~\ref{lem: HEAD bound4} in Sec.~\ref{sec: gradient estimation} and Lem.~2 in~\cite{liu2022experimental}, we get the final theorem by organizing the constants.  $\hfill\qed$    
\end{comment}
The proof is in \arxiv{App.~\ref{app: proof_FW}}{\CR{our technical report \cite{arxiv}}}. Our algorithm is based  on the Frank-Wolfe variant from Bian et al. \cite{bian2017guaranteed}. Similar to Thm.~2 in \cite{liu2022experimental}, our guarantee involves gradient estimation through truncating and sampling, due to Poisson  arrivals (see Asm.~\ref{asm: poisson}). However, incorporating Gaussian sources, we need to also sample from the Gaussian distribution to ensure a polynomial-time estimator. This requires combining a sub-Gaussian norm bound \cite{Lectures36709S19} with the aforementioned truncating and sampling techniques. 

\subsection{Gradient Estimation}
\label{sec: gradient estimation}
We describe here how to produce an unbiased, polynomial-time estimator of our gradient $\widehat{\nabla U}$, which is accessed by Eq.~\eqref{FW_1}.
There are three challenges in computing the true gradient \eqref{eq: derivative}: (a) the outer sum involves infinite summation over $n_{\source}^{\learner} \in \naturals$; (b) the outer expectation involves an exponential sum in $|\sources|-1$; and (c) the inner expectation involves an exponential sum in $\|\vc{n}^{\learner}\|$. The last arises from Gaussian sources, which differs from gradient estimation in~\cite{liu2022experimental}; this requires the use of a different bound (see Lem.~\arxiv{\ref{lem: HEAD bound3} in the Appendix}{\CR{xxx in our technical report~\cite{arxiv}}}), as well as a decoupling argument (as features and arrivals are jointly distributed). %\ls{Are the techniques used to handle (a) and (b) the same as in~\cite{liu2022experimental}???}

To address these challenges, we (a) truncate the infinite summation while maintaining the quality of estimation through a Poisson tail bound: 
\begin{equation}
\label{eq: HEAD}
 \mathrm{HEAD}_{\source, \type}^{\p}(n') = \sum_{n=0}^{n'} \Delta_{\source}^{\learner}(\vc{\lambda}^\learner,n) \cdot \prob[n_{\source}^{\learner}= n] \cdot T,
\end{equation}
where $\type = \type^{\learner}$, $l$ is the last node of $\p$, $\Delta_{\source}^{\learner}(\vc{\lambda}^\learner,n)$ is defined in Eq.~\eqref{eq: delta}, and $n'$ is the truncating parameter.
We then (b) sample $\vc{n}^{\learner}$ ($N_1$ samples) according to the Poisson distribution, parameterized by $\vc{\lambda}^{\learner} T$; and (c) sample $\vc{X}^{\learner}$ ($N_2$ samples) according to the Gaussian distribution. When $n'\geq \lambda_{\source}^{\learner} T$, we estimate the gradient by polynomial-time sampling:
\begin{align}
\label{eq: estimator}
  \widehat{\frac{\partial U}{\partial \lambda_{\source, \type}^{\p}}} =  \sum_{n=0}^{n'}    \widehat{\Delta_{\source}^{\learner}(\vc{\lambda}^\learner,n)} \cdot \prob[n_{\source}^{\learner}= n] \cdot T,
\end{align}
where
\begin{equation*}
\begin{split}
     \widehat{\Delta_{\source}^{\learner}(\vc{\lambda}^\learner,n)} = \frac{1}{N_1 N_2} \sum_{j=1}^{N_1} \sum_{k=1}^{N_2} (G^{\learner}(\vc{X}^{\learner,j,k},\vc{n}^{\learner,j}|_{n_{\source}^{\learner,j} = n+1}) - \\
    G^{\learner} (  \vc{X}^{\learner,j,k},\vc{n}^{\learner,j}|_{n_{\source}^{\learner,j} = n})),
\end{split}
\end{equation*}
$\vc{n}^{j}|_{n_{\source}^{\learner,j} = n}$ indicates vector $\vc{n}^{j}$ with $n_{\source}^{\learner,j} = n$, and $N_1$, $N_2$ are sampling parameters. At each iteration, we generate $N_1$ samples $\vc{n}^{\learner,j}$, $j = 1,\dots,N_1$ of the random vector $\vc{n}^\learner$ according to the Poisson distribution in Eq.~\eqref{eq:Poisson prod}, parameterized by the current solution vector $\vc{\lambda}^\learner T$. Having a sample $\vc{n}^{\learner,j}$, we could sample $N_2$ samples $\vc{X}^{\learner,j,k}$, $k = 1,\dots,N_2$ of random matrix $\vc{X}^{\learner,j} = [[\x^{\learner}_{\source,i}]_{i=1}^{n^{\learner,j}_{\source}}]_{\source \in \sources}$ according to the Gaussian distribution in Eq.~\eqref{eq:Gaussian prod}.
We bound the distance between the estimated and true gradient as follows:
\begin{lemma}
\label{lem: HEAD bound4}
For any $\delta \in (0,1)$, and $n'\geq \lambda_{\source}^{\learner} T$,
\begin{equation*}
\begin{split}
  - \gamma \max_{\learner \in \learners, \source \in \sources} \log(1 + \frac{\lambda_{\mathrm{MAX}}(\vc{\Sigma}_0^\learner)c_{\source}^2}{\sigma_{\source,\type}^2} )  \le \frac{\partial U}{\partial\lambda_{\source, \type}^{\p}} - \widehat{\frac{\partial U}{\partial\lambda_{\source, \type}^{\p}} } \le \\
 \gamma \max_{\learner \in \learners, \source \in \sources} \log(1 + \frac{\lambda_{\mathrm{MAX}}(\vc{\Sigma}_0^\learner)c_{\source}^2}{\sigma_{\source,\type}^2} ) + \prob[n^\learner_{\source}\geq n'+1] \frac{\partial U}{\partial\lambda_{\source, \type}^{\p}},
\end{split}
\end{equation*}
with probability greater than $1-2\cdot e^{-\gamma^2 N_1 N_2/2T^2(n'+1)} - |\sources|n'\delta - (|\sources|-1)\delta_{\source, \type}^{\p}$, where $\lambda_{\mathrm{MAX}}(\vc{\Sigma}_0^\learner)$ is the maximum eigenvalue of matrix $\vc{\Sigma}_0^\learner$.
\end{lemma}
\begin{comment}
\noindent\textbf{Proof Sketch.} Four important theorems are used for proof: Chernoff bounds described by Theorem A.1.16 in~\cite{alon2004probabilistic}, Sub-Gaussian bounds described by Theorem 8.3~\cite{Lectures36709S19}, law of total probability theorem, and Lemma 4 in \cite{liu2022experimental}. $\hfill\qed$
\end{comment}

\fussy
The proof is in \arxiv{App.~\ref{app: proof_HEAD bound4}}{\CR{our technical report \cite{arxiv}}}. Combining this estimated gradient in Eq.~\eqref{eq: estimator} with classic Frank-Wolfe variant \cite{bian2017guaranteed}, we propose Alg.~\eqref{FW} and establish Thm.~\ref{thm: FW}.

\section{Distributed Algorithm}
\label{sec: distributed algorithm}
Implementing Alg. \eqref{FW} in our distributed learning network is hard, as it requires the full knowledge of the network.
We thus present our distributed algorithm for solving Prob. \eqref{prob: utility}. The algorithm performs a \emph{primal dual gradient algorithm} over a modified Lagrangian to effectively find direction $\vc{v}$, defined in Eq.~\eqref{FW_1}, in a distributed fashion. The linearity of Eq.~\eqref{FW_1} ensures convergence, while Thm.~\ref{thm: FW} ensures the aggregate utility attained in steady state is within an $1 - \frac{1}{e}$ factor from the optimal. 

\subsection{Algorithm Overview}
Solving Prob.~\eqref{prob: utility} in a distributed fashion requires decentralizing Eqs.~\eqref{FW_1} and \eqref{FW_2}. Decentralizing the latter is easy, as Eq. \eqref{FW_2} can be executed across sources via:
\begin{align}
    \label{eq: update2}
    \lambda_{\source, \type}^{\p}(k + 1) = \lambda_{\source, \type}^{\p}(k) + \gamma v_{\source, \type}^{\p}(k),
\end{align}
across all $\source \in \sources$, and for all $\type \in \types$, $\p \in \paths_{\source, \type}$. 
We thus turn our attention to decentralizing Eq.~\eqref{FW_1}.

Eq.~\eqref{FW_1} is a linear program. Standard primal-dual distributed algorithms (see, e.g.,~\cite{srikant2004mathematics, lun2006minimum}) typically require strictly concave objectives, as they otherwise would yield to 
harmonic oscillations and not converge to an optimal point~\cite{feijer2010stability} in linear programs. An additional challenge arises from the multicast constraints in Eqs.~\eqref{cons: link capacity} and \eqref{cons: source}: the $\max$ function is non-differentiable. The maximum could be replaced by a set of multiple inequality constraints, but this approach does not scale well, introducing a new dual variable per additional constraint.

To address the first challenge, we follow Feijer and Paganini~\cite{feijer2010stability} and replace constraints of the form $u\leq 0$ with $\phi(u)\leq 0$, where $\phi(u) = e^u - 1$. In order to obtain a scalable differentiable Lagrangian, we use the approach in  \cite{srikant2004mathematics,lun2006minimum}: we replace the multicast constraints \eqref{cons: link capacity} and \eqref{cons: source} by 
\begin{equation}
\label{cons: link capacity approx}
\sum_{\source \in \sources, \type \in \types}\left( \sum_{\p \in \paths_{\source, \type}: e \in \p} (v_{\source,\type}^{\p})^\theta \right)^{\frac{1}{\theta}} \le \mu^e,
\end{equation}
for each link $e\in\edges$, and
\begin{equation}
\label{cons: source approx}
     \left( \sum_{\p \in \paths_{\source, \type}: e \in \p} (v_{\source,\type}^{\p})^\theta \right)^{\frac{1}{\theta}} \le \lambda_{\source, \type},
\end{equation}
for each source $\source\in \sources$ and each type $\type \in \types$. Note that this is tantamount to approximating $\|\cdot\|_\infty$ with $\|\cdot\|_\theta$. %; the latter indeed converges to $\|\cdot\|_\infty$ as $\theta\to\infty$, making the approximation arbitrarily accurate.
Combining these two approaches together, the Lagrangian for the modified problem is:
\begin{equation*}
\begin{split}
    & L(\vc{v}, \vc{q},\vc{r}, \vc{u}) = \langle \vc{v}, \widehat{\nabla U(\vc{\lambda})} \rangle - \sum_{e \in \edges} q_e ( e^{ g_e(\vc{v}) } -1 ) -
    \sum_{\source \in \sources, \type \in \types} \\
    & r_{\source, \type} ( e^{ g_{\source, \type} }(\vc{v}) - 1 ) - \sum_{\source \in \sources, \type \in \types} \sum_{\p \in \paths_{\source, \type}} u_{\source,\type}^{\p} ( e^{ g_{\source, \type}^{\p}(\vc{v}) } -1 ),
\end{split}
\end{equation*}
where 
\begin{align*}
   & g_{e}(\vc{v}) = \sum_{\source \in \sources, \type \in \types} \left( \sum_{\p \in \paths_{\source, \type}: e \in \p} (v_{\source,\type}^{\p})^\theta \right)^{\frac{1}{\theta}} - \mu^e, \\
   & g_{\source, \type}(\vc{v}) = \sum_{\p \in \paths_{\source, \type}} \left( \sum_{\p \in \paths_{\source, \type}: e \in \p} (v_{\source,\type}^{\p})^\theta \right)^{\frac{1}{\theta}} - \lambda_{\source, \type}, \\
   & g_{\source, \type}^{\p}(\vc{v}) = -v_{\source, \type}^{\p},
\end{align*}
and $\vc{q} = [q_e]_{e \in \edges}$, $\vc{r} = [r_{\source, \type}]_{\source \in \sources, \type \in \types}$, and $\vc{u} = [u_{\source,\type}^{\p}]_{\source \in \sources, \type \in \types, \p\in\paths_{\source, \type}}$ are non-negative dual variables. Intuitively, $L$ penalizes the infeasibility of network constraints. 

We apply a primal dual gradient algorithm over this modified Lagrangian to decentralize  Eq.~\eqref{FW_1}. %The algorithm iteratively maximizes primal variable $\vc{v}$ by gradient ascent on the modified Lagrangian, and minimizes dual variables $\vc{q}$, $\vc{r}$ and $\vc{u}$ by gradient descent.
In particular, at iteration $\tau + 1$, the primal variables are adjusted via \blu{gradient ascent}:
\begin{equation}
\label{eq: primal_update}
    v_{\source, \type}^{\p}(\tau+1) = v_{\source, \type}^{\p}(\tau) + m_{\source, \type}^{\p} \nabla_{v_{\source, \type}^{\p}} L(\tau),
\end{equation}
and the dual variables are adjusted via \blu{gradient descent}:
\begin{subequations}
\label{eq: dual_update}
\begin{align}
    & q_e (\tau + 1) = q_e(\tau) - k_e ( \nabla_{q_e} L(\tau) )_{q_e(\tau)}^+, \\
    & r_{\source, \type} (\tau+1) = r_{\source, \type}(\tau) - h_{\source, \type} ( \nabla_{r_{\source, \type}} L(\tau) )_{q_e(\tau)} )_{r_{\source, \type}(\tau)}^+, \\
    & u_{\source,\type}^{\p}(\tau+1) = u_{\source,\type}^{\p}(\tau) - w_{\source,\type}^{\p} ( \nabla_{u_{\source,\type}^{\p}} L(\tau) )_{q_e(\tau)} )_{u_{\source,\type}^{\p}(\tau)}^+,
\end{align}
\end{subequations}
for each edge $e \in \edges$, source $\source \in \sources$, type $\type \in \types$, and path $\p\in\paths_{\source, \type}$, where $k_e > 0$, $h_{\source, \type} > 0$, $w_{\source,\type}^{\p} > 0$, and $m_{\source, \type}^{\p} > 0$ are stepsize for $q_e$, $r_{\source, \type}$, $u_{\source,\type}^{\p}$ and $v_{\source, \type}^{\p}$, respectively, and
%\begin{align}
 $   (y)_x^+ = 
    \begin{cases}
    y, & x > 0, \\
    \max (y, 0), & x \le 0.
    \end{cases}$
%\end{align}
These operations can indeed be distributed across the network, as we describe in Sec.~\ref{sec: DFW}. The following theorem states the convergence of this modified primal dual gradient algorithm, \CR{according to Thm.~11 in \cite{feijer2010stability}}:
\begin{theorem}
\label{thm: PrimalDualConvergency}
The trajectories of the modified primal–dual gradient algorithm (Eqs.~\eqref{eq: primal_update} and \eqref{eq: dual_update}), with constant stepsize, converge to $\vc{v}^*_{\theta}$. The $\vc{v}^*_{\theta}$ is an optimum of Prob.~\eqref{FW_1} over $\feasibleset_{\theta}$, where $\feasibleset_{\theta}$ is $\feasibleset$ with \eqref{cons: link capacity} replaced by~\eqref{cons: link capacity approx}.
\end{theorem}
For $\vc{v}^*$ be the optimum of Eq. \eqref{FW_1}, $\lim_{\theta\to\infty} \vc{v}^*_{\theta} \to \vc{v}^*$, as $\feasibleset_{\theta} \to \feasibleset$.
Thus, our distributed algorithm preserves a $1-\frac{1}{e}$ approximation factor from the optimal objective value as stated in Thm.~\ref{thm: FW}, for large enough $\theta$.

\begin{comment}
\CR{\begin{cor}
    With probability greater than $1-\epsilon_0$, the output solution $\vc{\lambda}(K) \in \mathcal{D}$  of Alg.~\ref{alg: distributeFW} satisfies:
\begin{equation*}
    U(\vc{\lambda}(K)) \ge (1 - e^{\epsilon_1-1})  \max_{\vc{\lambda}\in\mathcal{D}}U(\vc{\lambda}) - \epsilon_2,
\end{equation*}
where $\epsilon_0$, $\epsilon_1$, and $\epsilon_2$ are defined in Thm.~\ref{thm: FW}.
\end{cor}
}
\end{comment}

\newcounter{tableeqn}
\renewcommand{\thetableeqn}{\thetable.\arabic{tableeqn}}
\newcounter{tablesubeqn}[tableeqn]
\renewcommand{\thetablesubeqn}{\thetableeqn\alph{tablesubeqn}}

\begin{table*}[t]
% \footnotesize
\centering
\stepcounter{table}% for \thetable
\begin{tabular}{lr}
% \hline
$v_{\source, \type}^{\p}(\tau+1) = v_{\source, \type}^{\p}(\tau) + m_{\source, \type}^{\p} \left( \nabla_{\lambda_{\source, \type}^{\p}} U(\vc{\lambda}) - \sum_{e \in \p} q_e(\tau) \cdot \right.
e^{ \sum_{\source' \in \sources, \type' \in \types} (v_{\source', \type'}^e(\tau))^{\frac{1}{\theta}} - \mu^e }
( v_{\source,\type}^e(\tau) )^{\frac{1-\theta}{\theta}} (v_{\source,\type}^{\p}(\tau))^{\theta-1} - $ &
\multirow{2}{*}{\refstepcounter{tableeqn} (\thetableeqn)\label{eq: primal_update2}} \\
$r_{\source,\type} e^{ \left( \sum_{\p \in \paths_{\source, \type}} (v_{\source,\type}^{\p}(\tau) )^\theta \right)^{\frac{1}{\theta}} - \lambda_{\source,\type}} \left( \sum_{\p \in \paths_{\source, \type}} (v_{\source,\type}^{\p}(\tau) )^\theta\right)^{\frac{1-\theta}{\theta}} 
\left. (v_{\source,\type}^{\p}(\tau))^{\theta-1} + u_{\source,\type}^{\p} \exp (- v_{\source, \type}^{\p}(\tau)) \right)$. \\
 
$v_{\source, \type}^e(\tau) = \sum_{\p \in \paths_{\source, \type}: e \in \p} (v_{\source,\type}^{\p}(\tau))^\theta$. &
\refstepcounter{tableeqn} (\thetableeqn)\label{eq: flow} \\

$q_e (\tau + 1) = q_e(\tau) + k_e \big( e^{ \sum_{\source \in \sources, \type \in \types} (v_{\source, \type}^e(\tau))^{\frac{1}{\theta}} \! - \! \mu^e } \! - \! 1 \big)_{q_e(\tau)}^+$. &
\refstepcounter{tableeqn} (\thetableeqn)\label{eq: dual_update4} \\
 
$r_{\source, \type} (\tau+1) = r_{\source, \type}(\tau) + h_{\source, \type} \left( e^{ \left( \sum_{\p \in \paths_{\source, \type}} (v_{\source,\type}^{\p}(\tau) )^\theta \right)^{\frac{1}{\theta}} - \lambda_{\source,\type}} - 1 \right)_{r_{\source, \type}(\tau)}^+$. &
\refstepcounter{tableeqn} (\thetableeqn)\label{eq: dual_update5} \\
 
$u_{\source,\type}^{\p}(\tau+1) = u_{\source,\type}^{\p}(\tau) + w_{\source,\type}^{\p} \left( e^{-v_{\source,\type}^{\p} (\tau)}) - 1 \right)_{u_{\source,\type}^{\p}(\tau)}^+$. &
\refstepcounter{tableeqn} (\thetableeqn)\label{eq: dual_update6} \\
% \hline
\end{tabular}
\addtocounter{table}{-1}%
\caption{\blu{Expanding primal and dual steps in Eqs.~\eqref{eq: primal_update} and \eqref{eq: dual_update}, so that we can execute FW algorithm distributively.}}\label{expand}
\end{table*}

\subsection{Distributed FW Implementation Details}
\label{sec: DFW}
We conclude by giving the full implementation details of the distributed FW algorithm and, in particular, the primal dual steps, describing the state maintained by every node, the messages exchanged, and the constituent state adaptations. 
Starting from $\vc{\lambda}(0) = \vc{0}$, the algorithm iterates over:
\begin{packed-enumerate}
    \item Each source node $\source \in \sources$ finds direction $v_{\source,\type}^{\p}(k)$ for all $\type \in \types$ and $\p\in\paths_{\source, \type}$ by the primal dual gradient algorithm.
    \item Each source node $\source$ updates $\lambda_{\source, \type}^{\p}(k + 1)$ using $v_{\source,\type}^{\p}(k)$ for all $\type \in \types$, and $\p\in\paths_{\source, \type}$ by executing Eq.~\eqref{eq: update2}.
\end{packed-enumerate}

\begin{algorithm}[t]
% \small
\caption{Distributed Frank-Wolfe Variant}
\label{alg: distributeFW}
\LinesNumbered
\KwIn{ $U:\feasibleset\to\reals_+$,  $\feasibleset$, stepsize $\delta\in(0,1]$.}
$\vc{\lambda}^0=0, \eta=0, k=0$ \\
\ForEach{source $\source \in \sources$}{
\While{$\eta<1$}{
    find direction $v_{\source,\type}^{\p}(k)$ by Alg. \ref{alg: PrimalDual} \\
    $\gamma_k = \min \{ \delta, 1-\eta \}$ \\
    $\lambda_{\source, \type}^{\p}(k + 1) = \lambda_{\source, \type}^{\p}(k) + \gamma v_{\source, \type}^{\p}(k)$, $\eta=\eta+\gamma_k$, $k=k+1$
}}
\Return $\vc{\lambda}(K)$
\end{algorithm}

\begin{algorithm}[t]
% \small
\caption{Primal Dual Gradient Algorithm}
\label{alg: PrimalDual}
\LinesNumbered
\KwIn{Rates $\vc{\lambda}$.}
\KwOut{Directions $\vc{v}$.}
Initialize direction $\vc{v}(0)=\vc{0}$, dual variables $\vc{q}(0),\vc{r}(0),\vc{u}(0) = \vc{0}$. \\
\ForEach{learner $\learner \in \learners$}{
    Send control messages carrying $\widehat{\nabla_{\lambda_{\source, \type}^{\p}} U(\vc{\lambda}^{\learner})}$ calculated by Eq. \eqref{eq: estimator} downstream over $\p$.
}
\For{$\tau=1,2,...$}{
    \ForEach{source $\source \in \sources$}{
        Generate features carrying $v_{\source,\type}^{\p}$ upstream.
    }
    \ForEach{edge $e \in \edges$}{
        Calculate $v_{\source,\type}^e$ using fetched $v_{\source,\type}^{\p}$ by Eq. \eqref{eq: flow}.
    }
    \ForEach{learner $\learner \in \learners$}{
        Send control messages downstream and collect $q_e$ and $v_{\source,\type}^e$ from traversed edges.
    }
    \ForEach{edge $e \in \edges$}{
        Update $q_e$ using calculated $v_{\source,\type}^e$ by Eq. \eqref{eq: dual_update4}.
    }
    \ForEach{source $\source \in \sources$}{
        Update $r_{\source,\type}$ using maintained $v_{\source,\type}^{\p}$ by Eq. \eqref{eq: dual_update5}. \\
        Update $u_{\source,\type}^{\p}$ using maintained $v_{\source,\type}^{\p}$ by Eq. \eqref{eq: dual_update6}. \\
        Update $v_{\source,\type}^{\p}$ using received $\widehat{\nabla_{\lambda_{\source, \type}^{\p}} U(\vc{\lambda}^{\learner})}$, $q_e$, $v_{\source,\type}^e$, and maintained $v_{\source,\type}^{\p}$ by Eq. \eqref{eq: primal_update2}.
    }

}
\Return $v_{\source,\type}^{\p}$ from each source $\source$
\end{algorithm}

We describe the first step in more detail. Every edge $e \in \edges$ maintains (a) Lagrange multiplier $q_e$, and (b) auxiliary variable $v_{\source, \type}^e$ for all $\source \in \sources$, $\type \in \types$. 
Every source $\source \in \sources$ maintains (a) direction $v_{\source,\type}^{\p}$ and (b) Lagrange multipliers $u_{\source,\type}^{\p}$, for all $\type \in \types$, $\p \in \paths_{\source, \type}$, and $r_{\source, \type}$, for all $\type \in \types$. The algorithm initializes all above variables by $\vc{0}$. Given the rates $\vc{\lambda}^{\learner}$, each learner estimates the gradient $\widehat{\nabla_{\lambda_{\source, \type}^{\p}} U(\vc{\lambda}^{\learner})}$ by Eq.~\eqref{eq: estimator}. Control messages carrying $\widehat{\nabla_{\lambda_{\source, \type}^{\p}} U(\vc{\lambda}^{\learner})}$ are generated and propagated over the path $\p$ in the reverse direction to sources.
Note that this algorithm is a synchronous algorithm where information needs to be exchanged within a specified intervals. 
Thus, the algorithm proceeds as follows during iteration $\tau+1$.
\begin{packed-enumerate}
    \item When feature $\x$ is generated from source $\source \in \sources$, it is propagated over the path $p$ to learner carrying direction $v_{\source,\type}^{\p}$. Every time it traverses an edge $e \in \edges$, edge $e$ fetches $v_{\source,\type}^{\p}$.

    \item After fetching all  $v_{\source,\type}^{\p}$, each edge $e \in \edges$ calculates the auxiliary variables $v_{\source, \type}^e(\tau)$
    for all $\source \in \sources$ and $\type \in \types$ by executing~\eqref{eq: flow}.
    
    \item  Learner $\learner \in \learners$ generates a control message,  sent over path $p$ in the reverse direction until reaching the source. When traversing edge $e \in \edges$, the control message collects $q_e$ and $v_{\source,\type}^e$. The source obtains these $q_e$ and $v_{\source,\type}^e$.
    
    \item After receiving all control messages, the edge $e \in \edges$ updates the Lagrangian multiplier $q_e (\tau + 1)$ using calculated $v_{\source,\type}^e$ by executing~\eqref{eq: dual_update4}.

    \item Upon obtaining $\widehat{\nabla_{\lambda_{\source, \type}^{\p}} U(\vc{\lambda}^{\learner})}$, $q_e$ and $v_{\source,\type}^e$, the source updates the Lagrangian multiplier $r_{\source, \type} (\tau+1)$ %using maintained $v_{\source,\type}^{\p}$
    by executing~Eq.~\eqref{eq: dual_update5},
     for all $\type \in \types$, 
    updates $u_{\source,\type}^{\p}(\tau+1)$ %using maintained $v_{\source,\type}^{\p}$ 
    by executing~Eq. \eqref{eq: dual_update6},
    for all $\type \in \types$, $\p\in\paths_{\source, \type}$,
    and updates the direction $v_{\source, \type}^{\p}(\tau+1)$ using the received $\widehat{\nabla_{\lambda_{\source, \type}^{\p}} U(\vc{\lambda}^{\learner})}$ $q_e$, $v_{\source,\type}^e$ %, and maintained $v_{\source,\type}^{\p}$ 
    by executing~Eq. \eqref{eq: primal_update2}, 
    for all $\type \in \types$, $\p\in\paths_{\source, \type}$.
\end{packed-enumerate}
\CR{The above steps are summarized in Alg.~\ref{alg: distributeFW}. }
This implementation is indeed in a decentralized form: the updates happening on sources and edges require only the knowledge of entities linked to them.

\section{Projected Gradient Ascent}
\label{sec: extensions}

We can also solve Prob.~\eqref{prob: utility} by  \emph{projected gradient ascent} (PGA)  \cite{hassani2017gradient}. Decentralization reduces then to a primal-dual algorithm~\cite{srikant2004mathematics, feijer2010stability} over a strictly convex objective, which is easier than the FW variant we studied; however, PGA comes with a worse approximation guarantee. We briefly outline  this  below.  Starting from $\vc{\lambda}(0) = \vc{0}$,
PGA iterates over:
\begin{subequations}
\begin{align}
    & \vc{v}(k) =  \vc{\lambda}(k) + \gamma \widehat{\nabla U(\vc{\lambda}(k))} \label{PGA_1}\\
    & \vc{\lambda}(k+1) = \Pi_{\feasibleset}(\vc{v}(k)) \label{PGA_2}
\end{align}
\end{subequations}
where $\widehat{\nabla U(\cdot)}$ is an estimator of the gradient $\nabla U$, $\gamma$ is the stepsize, and $\Pi_{\feasibleset}(\x)=\arg\min_{\vc{y}\in \feasibleset} (\vc{y}-\x)^2$ is the orthogonal projection. Our gradient estimator in Sec.~\ref{sec: gradient estimation} would again be used here to compute $\widehat{\nabla U(\vc{\lambda}^k)} $. Note that, to achieve the same quality of gradient estimator, PGA usually takes a longer time compared to the FW algorithm. This comes from larger $\lambda^\learner_{\source}(k)$, thus, larger truncating parameter $n'$, during the iteration $k$ (see also Sec.~\ref{sec: competitor algorithms} for how we set algorithm parameters). Furthermore, PGA comes with a worse approximation guarantee compared to the FW algorithm, namely, $1/2$ instead of $1-1/e\approx 0.63$; this would follow by combining the guarantee in~\cite{hassani2017gradient} with the gradient estimation bounds in Sec.~\ref{sec: gradient estimation}. 
Similar to distributed FW, we can easily decentralize Eq.~\eqref{PGA_1}. Eq.~\eqref{PGA_2} has a strictly convex objective, so we can directly decentralize it through a standard primal-dual algorithm with approximated multicast link capacity constraints, as in Eq.~\eqref{cons: link capacity approx}. Convergence then is directly implied by Thm.~5 in~\cite{feijer2010stability}.

\section{Numerical Evaluation}
\label{sec: numerical evaluation}

\begin{table}[t]
\small
\caption{Graph Topologies and Experiment Parameters}
\label{tab:topologies}
\centering
\setlength{\tabcolsep}{1.5mm}{
\begin{tabular}{ccccccccc}
Graph & $|V|$ & $|E|$ & $\mu^e$ & $|\learners|$ & $|\sources|$ & $|\types|$ & $U_{\texttt{DFW}}$ & $U_{\texttt{DPGA}}$ \\
\hline
\multicolumn{9}{c}{synthetic topologies}\\
\hline
\texttt{ER} & 100 & 1042 & 5-10 & 5 & 10 & 3 & 351.8 & 357.3 \\
\texttt{BT} & 341 & 680 & 5-10 & 5 & 10 & 3 & 163.3 & 180.6 \\
\texttt{HC} & 128 & 896 & 5-10 & 5 & 10 & 3 & 320.4 & 343.7 \\
\texttt{star} & 100 & 198 & 5-10 & 5 & 10 & 3 & 187.1 & 206.0 \\
\texttt{grid} & 100 & 360 & 5-10 & 5 & 10 & 3 & 213.6 & 236.9 \\
\texttt{SW} & 100 & 491 & 5-10 & 5 & 10 & 3 & 269.5 & 328.4 \\
\hline
\multicolumn{9}{c}{real backbone networks} \\
\hline
\texttt{GEANT} & 22 & 66 & 5-8 & 3 &  3 & 2 & 116.4 & 117.4 \\ 
\texttt{Abilene} & 9 & 26 & 5-8 & 3 &  3 & 2 & 141.3 & 139.6 \\
\texttt{Dtelekom}  & 68 & 546 & 5-8 & 3 & 3 & 2 & 125.1 & 142.3 \\
\end{tabular}}
\end{table}
\subsection{Experimental Setup}

\noindent\textbf{Topologies.} We perform experiments over five synthetic graphs, namely, Erd\H{o}s-R\'enyi (\texttt{ER}), balanced tree (\texttt{BT}), hypercube (\texttt{HC}), grid\_2d (\texttt{grid}), and small-world (\texttt{SW})~\cite{kleinberg2000small}, and three backbone network topologies: Deutsche Telekom (\texttt{DT}), \texttt{GEANT}, and \texttt{Abilene}~\cite{rossi2011caching}. The graph parameters of different topologies are shown in Tab.~\ref{tab:topologies}.

\fussy
\noindent\textbf{Network Parameter Settings.} For each network, we uniformly at random (u.a.r.) select $|\learners|$ learners and $|\sources|$ data sources. Each edge $e\in\mathcal{E}$ has a link capacity $\mu^e$ and types  $\types$ as indicated in Tab.~\ref{tab:topologies}.
Sources generate feature vectors with dimension $d=100$ within data acquisition time $T=1$. Each source $s$ generates the data $(\x, y)$ of type $\type$ label with rate $\lambda_{\source,\type}$, uniformly distributed over [5,8]. Features $\x$ from source $s$ are generated following a zero mean Gaussian distribution, whose covariance is generated as follows. First, we separate features into two classes: well-known and poorly-known. Then, we set the corresponding Gaussian covariance (i.e., the diagonal elements in  $\vc{\Sigma}_{\source}$) to low (uniformly from 0 to 0.01) and high (uniformly from 10 to 20) values, for well-known and poorly-known features, respectively. Source $s$ labels $y$ of type $t$ using ground-truth models, as discussed below, with Gaussian noise, whose variance $\sigma_{\source,\type}$ is chosen u.a.r. (uniformly at random)~from 0.5 to 1.
For each source, the paths set consists of the shortest paths between the source and every learner in $\learners$. 
Each learner has a target model $\vc{\beta}_{\type^\learner}$, which is sampled from a prior normal distribution as follows. Similarly to sources, we separate features into interested and indifferent. Then, we set the corresponding prior covariance (i.e., the diagonal elements in  $\vc{\Sigma}_0^\learner$) to low (uniformly from 0 to 0.01) and high (uniformly from 1 to 2) values, and set the corresponding prior mean to 1 and 0, for interested and indifferent features, respectively.

\noindent\textbf{Algorithms.}
\label{sec: competitor algorithms}
We implement our algorithm and several competitors. First, there are four centralized algorithms:
\begin{packed-itemize}
    \item \texttt{MaxTP}: This maximizes the aggregate incoming traffic rates (throughput) of learners, i.e.: 
    \begin{equation}
    \label{eq: MaxTP}
    \max_{\vc{\lambda} \in \feasibleset}: U_{\texttt{MaxTP}}(\vc{\lambda})=\sum_{\learner\in\learners}\sum_{\source\in\sources} \lambda^\learner_{\source}.
    \end{equation}
    \item \texttt{MaxFair}: This maximizes the aggregate $\alpha$-fair utilities~\cite{srikant2004mathematics} of the incoming traffic at learners, i.e.:
    \begin{equation}
    \label{eq: MaxFair}
    \max_{\vc{\lambda} \in \feasibleset}: U_{\texttt{MaxFair}}(\vc{\lambda})=\sum_{\learner\in\learners}(\sum_{\source\in\sources} \lambda^\learner_{\source})^{1-\alpha}/(1-\alpha).
    \end{equation}
    We set $\alpha = 2.$
    \item \texttt{FW}: This is Alg.~\eqref{FW}, as proposed in Sec.~\ref{sec: centralized algorithm}.
    \item \texttt{PGA}: This is the algorithm we proposed in Sec.~\ref{sec: extensions}. 
\end{packed-itemize}
We also implement their corresponding distributed versions: \texttt{DMaxTP}, \texttt{DMaxFair}, \texttt{DFW} (algorithm in Sec.~\ref{sec: DFW}), and \texttt{DPGA} (see Sec.~\ref{sec: extensions}). The objectives of \texttt{MaxTP} Eq. \eqref{eq: MaxTP} and \texttt{MaxFair} Eq. \eqref{eq: MaxFair} are linear and strictly concave, respectively. The modified primal dual gradient algorithm, used in \texttt{DFW}, and basic primal dual gradient algorithm, used in \texttt{DPGA}, directly apply to \texttt{DMaxTP} and \texttt{DMaxFair}, respectively.

\noindent\textbf{Algorithm Parameter Settings.} We run \texttt{FW}/\texttt{DFW} and \texttt{PGA}/\texttt{DPGA} for $K = 50$ iterations, i.e. stepsize $\gamma = 0.02$, with respect to the outer iteration. In each iteration, we estimate the gradient according to Eq. \eqref{eq: estimator} with sampling parameters $N_1 = 50$, $N_2 = 50$, and truncating parameters $n' = \max \{ \lceil 2\max_{\learner,\source}\lambda^\learner_{\source}T \rceil,$ $ 10 \}$, where $\lambda^\learner_{\source}$ is given by the current solution. We run the inner primal-dual gradient algorithm for 1000 iterations and set parameter $\theta=10$ when approximating the $\max$ function via Eqs.~\eqref{cons: link capacity approx} and \eqref{cons: source approx}. We compare the performance metrics (Aggregate utility and Infeasibility, defined in Sec.~\ref{sec: metrics}), between centralized and distributed versions of each algorithm under different stepsizes, and we choose the best stepsize for distributed primal-dual algorithms. We further  discuss   the impact of the stepsizes in Sec.~\ref{sec: experimental result}. %Note that, the most time consumption part of the algorithms is the gradient estimation.

\subsection{Performance Metrics}
\label{sec: metrics}
To evaluate the performance of the algorithms, we use the \emph{Aggregate Utility}, defined in Eq.~\eqref{eq: objective function} as one metric. Note that as the aggregate utility involves a summation with infinite support, we thus need to resort to sampling to estimate it; we set $N_1 = 100$ and $N_2 = 100$. Also, we define an \emph{Estimation Error} to measure the model learning/estimation quality. Formally, it is defined as:
%\begin{align}
%\label{eq: avg norm}
%   \begin{split} 
 $   \frac{1}{|\learners|} \sum_{\learner \in \learners} \frac{\| \hat{\vc{\beta}}_\map^{\learner} - \vc{\beta}^{\learner} \|}{\|\vc{\beta}^{\learner}\|},$
%    
    %&= \arg\min_{\beta} \frac{1}{2\sigma^2}||X_\mathcal{S}\beta-y||^2_2 + \frac{1}{2}(\beta-\beta_0)^{\mathrm{T}}\Sigma_0^{-1}(\beta-\beta_0)\notag\\
%    & =  \frac{1}{|\learners|} \sum_{\learner \in \learners} \| ( \vc{X}^{\learner\top} \vc{\Sigma}^{\learner -1} \vc{X}^{\learner} + \vc{\Sigma}_{\learner}^{-1})^{-1} \cdot \\
 %   & (\vc{X}^{\learner\top} \vc{\Sigma}^{\learner-1} \vc{y}^{\learner} + \vc{\Sigma}_{\learner}^{-1}\vc{\beta}_0^{\learner} )- \vc{\beta}^{\learner} \| / \|\vc{\beta}^{\learner}\|, 
%    \end{split}
%\end{align}
following the equation of MAP estimation Eq.~\eqref{eq:map}. We average over 2500 realizations of the number of data arrived at the learner $\{\vc{n}^\learner\}_{\learner \in \learners}$ and features $\{\vc{X}^\learner \}_{\learner \in \learners}$ , and 20 realizations of ground-truth models $\{\vc{\beta}^{\learner}\}_{\learner \in \learners}$. Finally, we define an \emph{Infeasibility} to measure the feasibility of solutions, as primal dual gradient algorithm used in distributed algorithms does not guarantee feasibility. It averages the total violations of constraints \eqref{eq: lambda}-\eqref{cons: source} over the number of constraints. 

\begin{figure}[t!]
     \centering
     \begin{subfigure}{1.\linewidth}
         \centering
         \includegraphics[width=\textwidth]{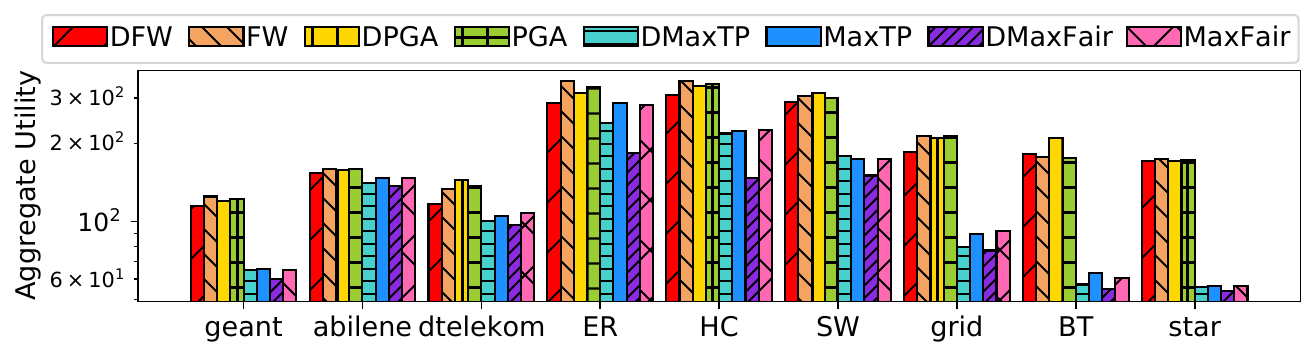}
         \caption{Aggregate Utility}
         \label{fig: aggregate utility}
     \end{subfigure}
     % \hfill
     \begin{subfigure}{1.\linewidth}
         \centering
         \includegraphics[width=\textwidth]{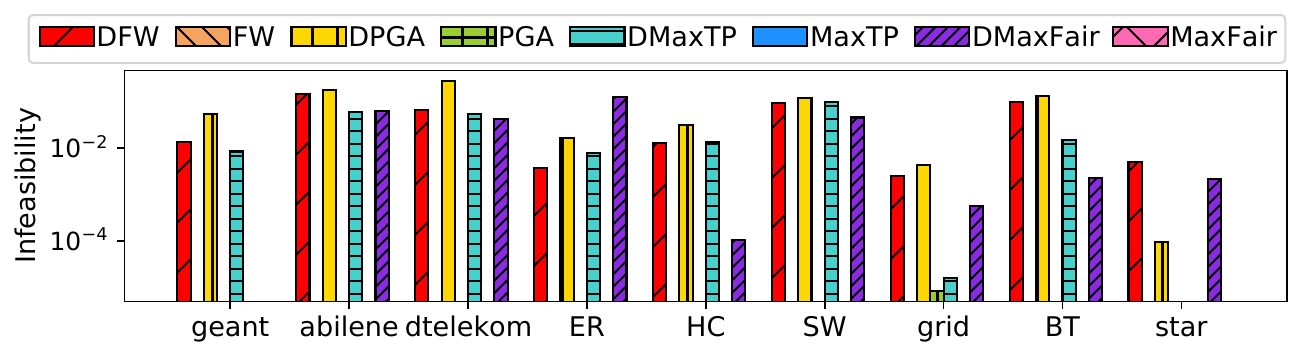}
         \caption{Infeasibility}
         \label{fig: infeasibility}
     \end{subfigure}
     \begin{subfigure}{1.\linewidth}
         \centering
         \includegraphics[width=\textwidth]{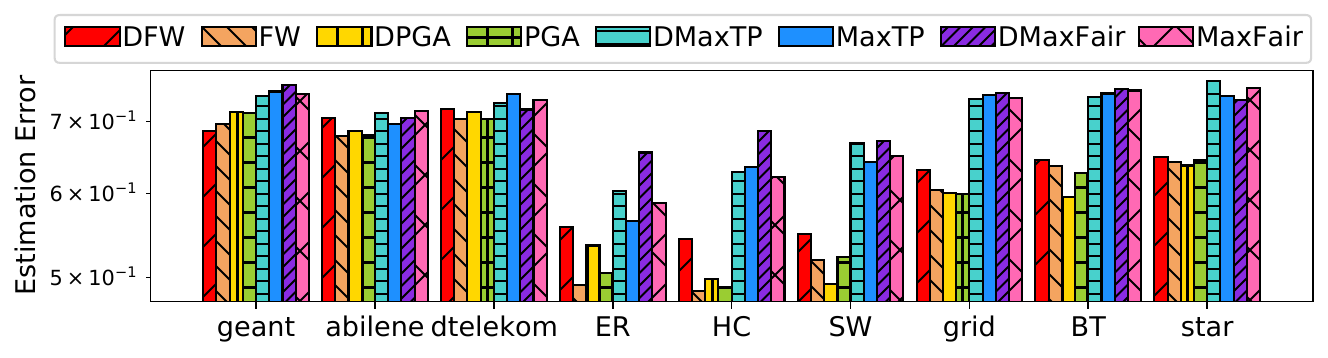}
         \caption{Estimation Error}
         \label{fig: estimation error}
     \end{subfigure}
        \caption{Aggregate utility, infeasibility and estimation error across networks. \texttt{DFW} and \texttt{DPGA} perform very well in terms of maximizing the utility and minimizing the estimation error in all networks. The aggregate utilities of \texttt{DFW} and \texttt{DPGA} are also listed in Tab.~\ref{tab:topologies}. Furthermore, their performances are close to their centralized versions: \texttt{FW} and \texttt{PGA}, with an acceptable infeasibility~$\sim 0.1$.}
        \label{fig: different_top}
\end{figure}

\begin{figure}[t!]
    \centering
  \includegraphics[width=1.\linewidth]{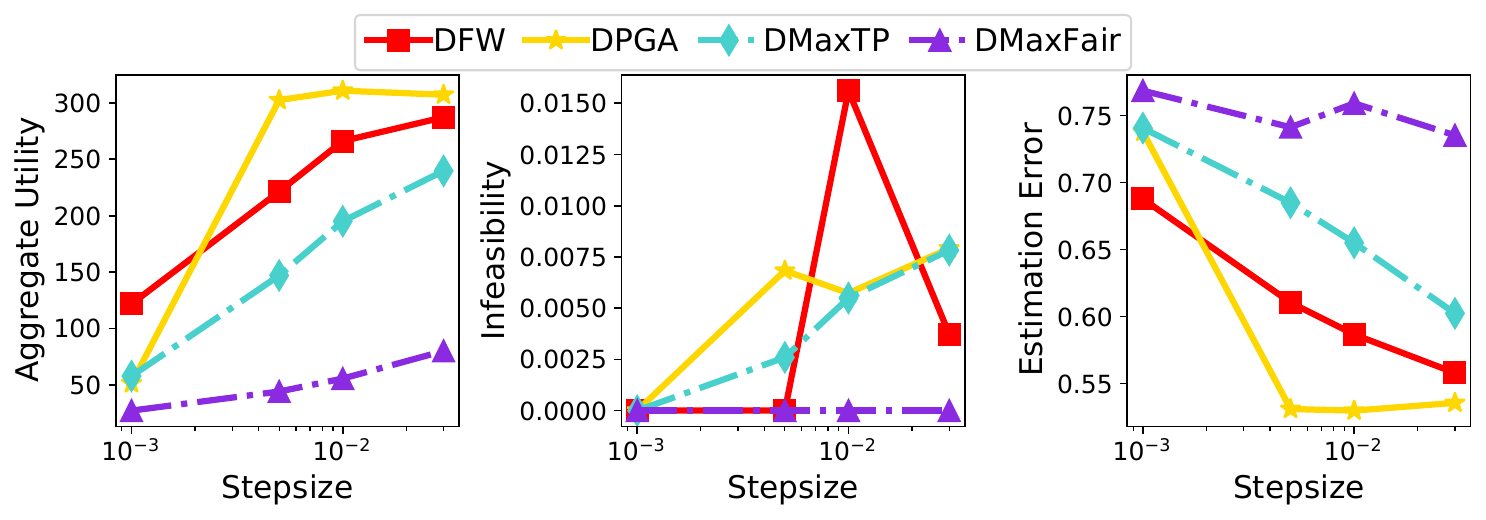}
    \caption{Stepsize effect on primal dual gradient algorithms over topology \texttt{ER}. Larger stepsizes lead to better performance, and \texttt{DFW} and \texttt{DPGA} are always the best in terms of both utility and estimation error. However, stepsizes above $0.03$ lead to numerical instability. }
    \label{fig: stepsize}
\end{figure}

\begin{figure}[t!]
     \centering
     \begin{subfigure}{1.\linewidth}
         \centering
         \includegraphics[width=\textwidth]{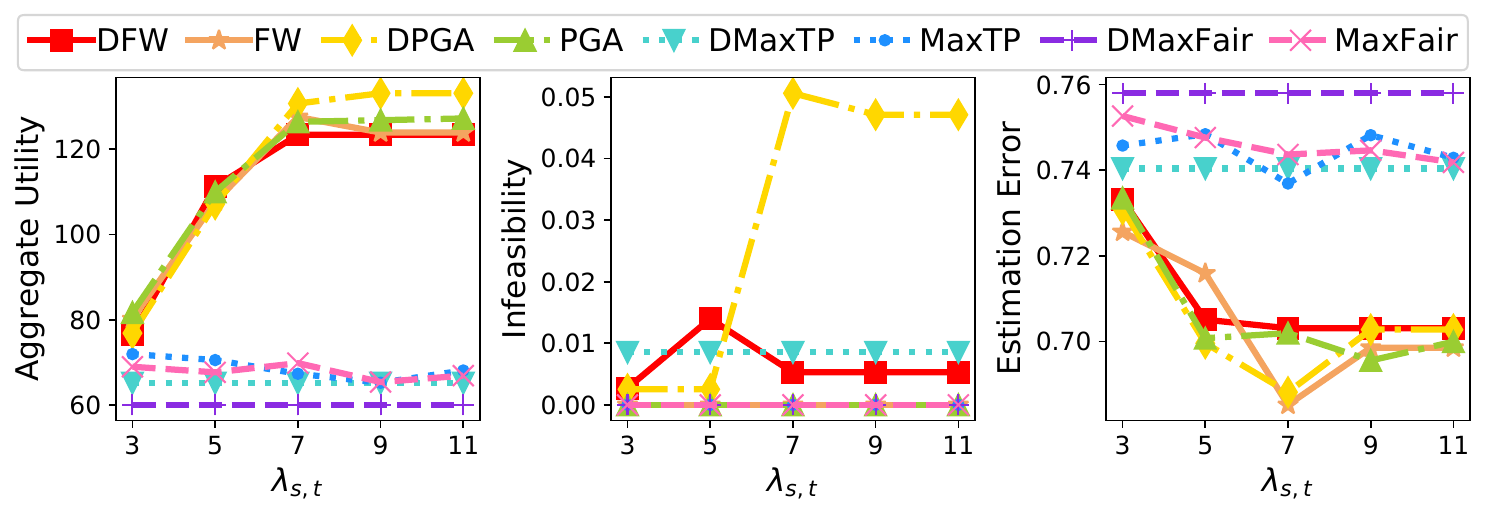}
         \caption{Varying source rate}
         \label{fig: rate}
     \end{subfigure}
     % \hfill
     \begin{subfigure}{1.\linewidth}
         \centering
         \includegraphics[width=\textwidth]{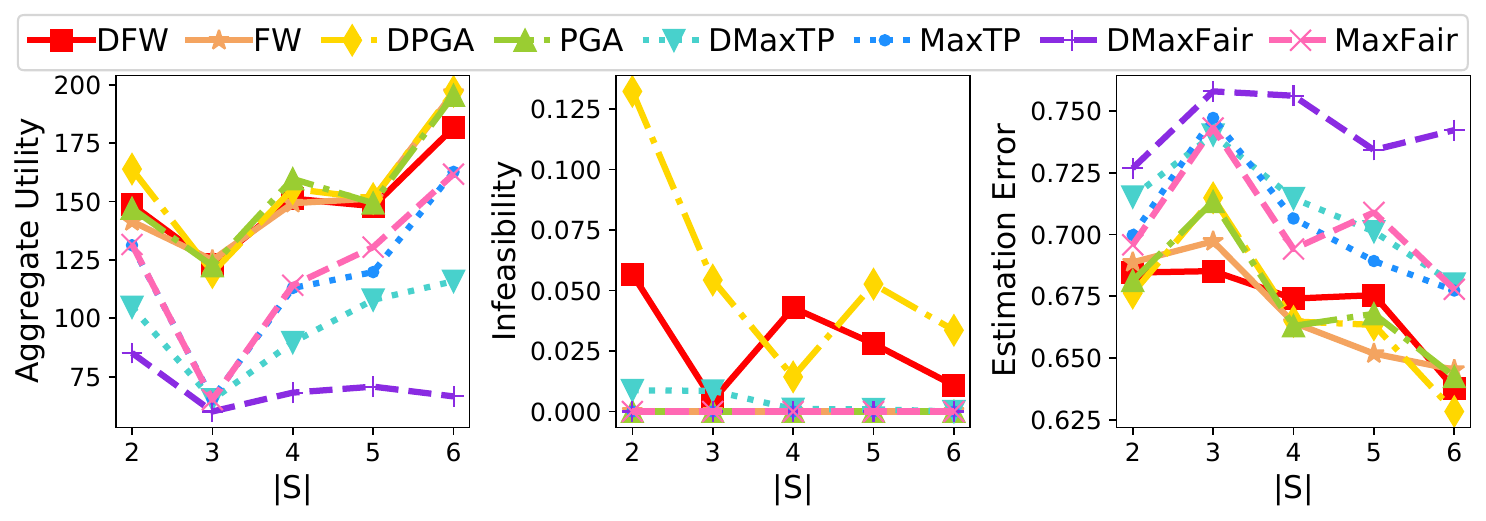}
         \caption{Varying source set size}
         \label{fig: source}
     \end{subfigure}
        \caption{Varying source rates and source set size over \texttt{GEANT}. When increasing source rates and source set sizes, learners receive more data. This leads to higher aggregate utility, and lower estimation error. Our algorithms, \texttt{DFW} and \texttt{DPGA}, stay close to their centralized versions (\texttt{FW} and \texttt{PGA}) and outperform competitors in both metrics, with a small change in feasibility.}
\end{figure}

\begin{figure}[t!]
    \centering
    \includegraphics[width=1.\linewidth]{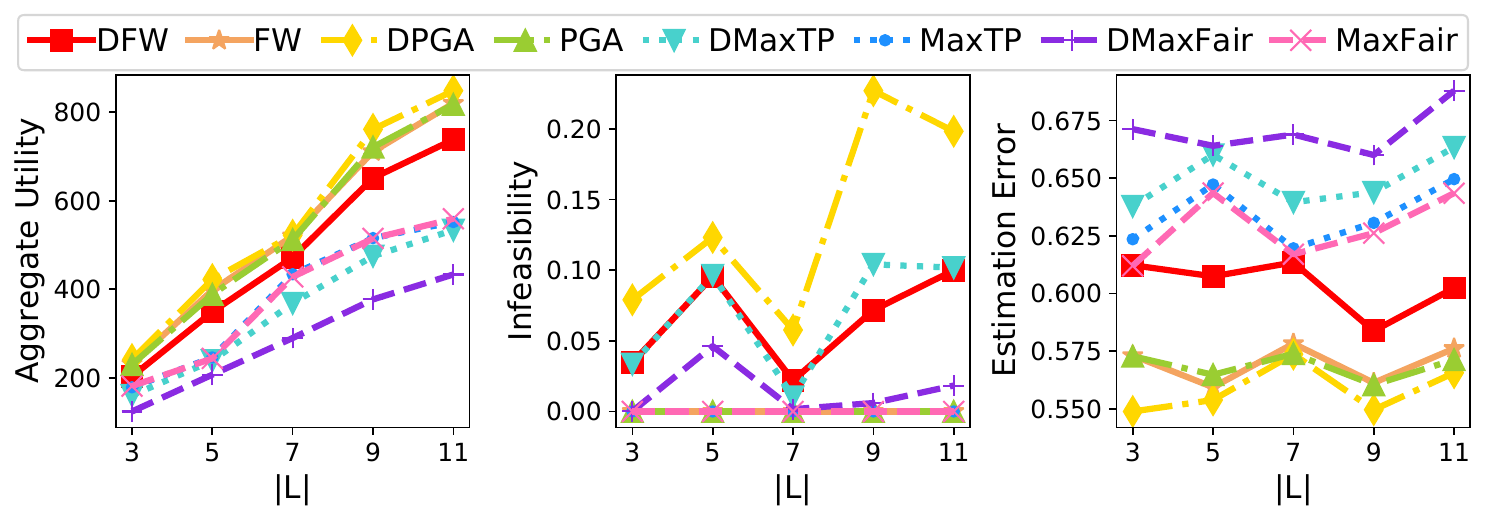}
    \caption{Varying learner set size over topology \texttt{SW}. The aggregate utility increases, while the estimation error remains essentially unchanged, as the number of learners increases. \texttt{DFW} and \texttt{DPGA} again stay close to their centralized versions and outperform competitors.}
    \label{fig: learner}
\end{figure}

\subsection{Results}
\label{sec: experimental result}
\noindent \textbf{Different Topologies. }
We first compare the proposed algorithms with several baselines in terms of aggregate utility, infeasibility and estimation error over several network topologies, shown in Fig.~\ref{fig: different_top}. Our proposed algorithms dramatically outperform all competitors, and our distributed algorithms perform closed to their corresponding centralized algorithms (c.f. Thm.~\ref{thm: PrimalDualConvergency}). All of the distributed algorithms have a low infeasibility, which is less than 0.1. Such violations over the constraints are expected by primal-dual gradient algorithms, since they employ soft constraints.

\fussy
\noindent\textbf{Effect of Stepsizes. }
We study the effect of the stepsizes in the primal dual gradient algorithms used in the distributed algorithms \texttt{DFW}, \texttt{DPGA}, \texttt{DMaxTP} and \texttt{DMaxFair} over \texttt{ER}, shown in Fig.~\ref{fig: stepsize}. When the algorithms are stable, larger stepsizes achieve better performance w.r.t.~both aggregate utility and estimation error, while worse performance with respect to~infeasibility. However, if the stepsize is too large, the algorithms do not converge and become numerically unstable. It is crucial to choose an appropriate stepsize for better performance, while maintaining convergence. In all convergent cases, \texttt{DFW} and \texttt{DPGA} outperform competitors.

\noindent\textbf{Varying Source Rates and Source Set Size. }
Next, we evaluate how algorithm performance is affected by varying the (common) source rates $\lambda_{\source, \type}$ over topology \texttt{GEANT}. As shown in Fig.~\ref{fig: rate}, when source rates increase, the aggregate utility first increases very fast and then tapers off. Higher source rates indicate more data received at the learners, hence the greater utility. However, due to DR-submodularity, the marginal gain decreases as the number of sources increases. Furthermore, under limited bandwidth, if link capacities saturate, there will be no further utility increase. The same interpretation applies to the estimation error. We observe similarly changing patterns when varying the source set size $|\sources|$ over topology \texttt{GEANT}, shown in Fig.~\ref{fig: source}, since more sources also indicates learners receive more data. However, the curve changes are not as smooth as those for increasing the source rates. This is because  varying source sets also changes the available paths, corresponding link bandwidth utilization, indexes of well-known features, etc. Overall, we observe that our algorithms, \texttt{DFW} and \texttt{DPGA}, stay close to their centralized versions (\texttt{FW} and \texttt{PGA}) and outperform competitors in both metrics, with a small change in feasibility.

\noindent\textbf{Varying Learner Set Size. }
Finally, we evaluate the effect of the learner set size~$|\learners|$ over topology \texttt{SW}. Fig.~\ref{fig: learner} shows that as  the number of learners increases so does  the aggregate utility, while the estimation errors essentially remain the same. With multicast transmissions, increasing the number of learners barely affects the amount of data received by each learner. Thus, the aggregate utility increases as expected, while the \emph{average} utility per learner (the aggregate utility divided by the number of learners) and, consequently, the estimation error hardly change, when more learners are in the network. Again, our algorithms, \texttt{DFW} and \texttt{DPGA}, stay close to their centralized versions (\texttt{FW} and \texttt{PGA}) and outperform competitors.

\section{Conclusion}
\label{sec: conclusion}
We generalize the experimental design networks by considering Gaussian sources and multicast transmissions. A poly-time distributed algorithm with $1-1/e$ approximation guarantee is proposed to facilitate heterogeneous model learning across networks.
One limitation of our distributed algorithm is its synchronization. It is natural to extend the model to an asynchronous setting, which better resembles the reality of large networks. One possible solution is that sources and links compute outdated gradients~\cite{low1999optimization}. Another interesting direction is to estimate gradients through shadow prices~\cite{ioannidis2009optimal,kelly1998rate}, instead of sampling. Furthermore, how a model trained by one learner benefits other training tasks in experimental design networks is also a worthwhile topic to study.
\CR{The authors have provided public access to their code and data.\footnote{\href{https://github.com/neu-spiral/DistributedNetworkLearning}{https://github.com/neu-spiral/DistributedNetworkLearning}}}

\section*{Acknowledgment}

\sloppy
\CR{The authors gratefully acknowledge support from the National Science Foundation (grants 1718355, 2106891, 2107062, and 2112471).}

\fussy

\bibliographystyle{IEEEtran}  
\bibliography{references}

\newpage
\onecolumn
\arxiv{
\appendices

\section{Proof of Theorem \ref{thm: submodular}} 
\label{app: proof_submodular}
%In general, linear combinations of submodular functions with positive weights are also submodular. However, proving the DR-submodularity of the objective poses a challenge precisely because the set of features $\vc{X}^{\learner}$ observed by a learner depends on $\vc{n}^{\learner}$, the number of samples received. As a result, such an ``averaging'' argument cannot be directly applied here.

%To resolve this, we
Consider an equivalent data generation  process, in which  infinite sequences of independent Gaussian variables $\widetilde{\vc{X}^{\learner}} = [[ \x_{\source,i}^{\learner} ]_{i = 1}^{\infty}]_{\source \in \sources}$ from different sources $\source$ are received by the learner, but only an initial part of each (as governed by $\vc{n}^{\learner}$) is observed. Note that the statistics of this process are identical to the actual arrival process, where samples are independent conditioned on $\vc{n}^{\learner}$.
Then, combining the proof of Lem.~1 in \cite{liu2022experimental} using Sylvester's determinant identity \cite{akritas1996various, horel2014budget} and telescoping sum, we get:
\begin{lemma}
Function $G^{\learner}(\widetilde{\vc{X}^{\learner}}, \vc{n}^{\learner})$ is (a) monotone-increasing and (b) DR-submodular w.r.t. $\vc{n}^{\learner}$, where $\widetilde{\vc{X}^{\learner}} = [[ \x_{\source,i}^{\learner} ]_{i = 1}^{\infty}]_{\source \in \sources}$.
\end{lemma}
\begin{proof}
For $\vc{n}\in\naturals^{|\sources|}$, we have:
\begin{align*}
     G(\vc{X}, \vc{n}+ \vc{e}_{\source'})) - G(\vc{X}, \vc{n})  = & \log\det\big( \sum_{\source \in \sources}\sum_{i=1}^{n_\source} \frac{\x_{\source,i}\x_{\source,i}^{\top}}{\sigma_{\source}^2} + \frac{\x_{\source',n_{\source'} + 1}\x_{\source',n_{\source'} + 1}^{\top}}{\sigma_{\source'}^2} + \vc{\Sigma}_0^{-1}\big) -\log\det\big( \sum_{\source \in \sources}\sum_{i=1}^{n_\source} \frac{\x_{\source,i}\x_{\source,i}^{\top}}{\sigma_{\source}^2} + \vc{\Sigma}_0^{-1}\big) \\
    = & \log\det\big(\vc{I}_d + \frac{\x_{\source',n_{\source'} + 1}\x_{\source',n_{\source'} + 1}^{\top}}{\sigma_{\source'}^2} \big( \sum_{\source \in \sources}\sum_{i=1}^{n_\source} \frac{\x_{\source,i}\x_{\source,i}^{\top}}{\sigma_{\source}^2} + \vc{\Sigma}_0^{-1}\big)^{-1} \big) \\
    = & \log\det \big( \vc{I}_1 + \frac{1}{\sigma_{\source'}^2} \x_{\source',n_{\source'} + 1}^{\top} \vc{A}(\vc{n}) \x_{\source',n_{\source'} + 1} \big) \\
    = & \log \big( 1 + \frac{1}{\sigma_{\source'}^2} \x_{\source',n_{\source'} + 1}^{\top} \vc{A}(\vc{n}) \x_{\source',n_{\source'} + 1} \big)
\end{align*}
where $\vc{A}(\vc{n}) =  \big( \sum_{\source \in \sources}\sum_{i=1}^{n_\source} \frac{\x_{\source,i}\x_{\source,i}^{\top}}{\sigma_{\source}^2} + \vc{\Sigma}_0^{-1}\big)^{-1}$ and the second last equality follows Sylvester's determinant identity. Then for $\vc{n}\in\naturals^{|\sources|}$ and $k\in\naturals$, by telescoping sum, we have
\begin{align*}
     & G(\vc{X}, \vc{n}+ k\vc{e}_{\source'})) - G(\vc{X}, \vc{n}) \\
    = & G(\vc{X}, \vc{n}+ k\vc{e}_{\source'})) - G(\vc{X}, \vc{n}+ (k-1)\vc{e}_{\source'})) + \cdots + (G(\vc{X}, \vc{n}+ \vc{e}_{\source'}))- G(\vc{X}, \vc{n})) \\
    = & \log \big( 1 + \frac{1}{\sigma^2} \x_{\source',n_{\source'} + k}^{\top} \vc{A}(\vc{n}+ (k-1)\vc{e}_{\source'}) \x_{\source',n_{\source'} + k} \big) + \cdots + \log \big( 1 + \frac{1}{\sigma^2} \x_{\source',n_{\source'} + 1}^{\top} \vc{A}(\vc{n}) \x_{\source',n_{\source'} + 1} \big)
\end{align*}
where $\vc{A}(\vc{n}) =  \big( \sum_{\source \in \sources}\sum_{i=1}^{n_\source} \frac{\x_{\source,i}\x_{\source,i}^{\top}}{\sigma_{\source}^2} + \vc{\Sigma}_0^{-1}\big)^{-1}$ and the second last equality follows Sylvester's determinant identity. The monotonicity of $G$ follows because $A(\vc{n})$ is positive semidefinite. Finally, since the matrix inverse is decreasing over the positive semi-definite order, we have $A(\vc{n})\succeq A(\vc{m})$, $\forall\ \vc{n}, \vc{m}\in\naturals^{|\sources|}, k \in \naturals \text{ and } \vc{n}\leq\vc{m}$, which leads to $G(\vc{X},\vc{n} + k\vc{e}_s) - G(\vc{X},\vc{n})\geq G(\vc{X}, \vc{m} + k\vc{e}_s) - G(\vc{X},\vc{m})
$.
\end{proof}

Armed with this result, we use a truncation argument to prove that taking expectations w.r.t.~the (infinite) sequences $\widetilde{\vc{X}^{\learner}} = [[ \x_{\source,i}^{\learner} ]_{i = 1}^{\infty}]_{\source \in \sources}$ preserves submodularity:
\begin{lemma}
\label{lem: submodular}
Function $g^{\learner}(\vc{n}^{\learner}) = \mathbb{E}_{\vc{X}^{\learner}} [G(\widetilde{\vc{X}^{\learner}},\vc{n}^{\learner})|\vc{n}^{\learner}]$ is (a) monotone-increasing and (b) DR-submodular w.r.t. $\vc{n}^{\learner}$.
\end{lemma}
\begin{proof}
Let's consider a projection: 
\begin{equation*}
    \proj_{n_0}(\vc{n}) = \tilde{\vc{n}} = [\tilde{n}_\source]_{\source\in \sources},
\end{equation*}
where
\begin{equation*}
    \tilde{n}_\source =
    \begin{cases}
    n_s,  \text{ if } n_s \le n_0 \\
    n_0,  \text{ otherwise.}
    \end{cases}
\end{equation*}
Then, consider a function: 
\begin{equation*}
G_{n_0}(\widetilde{\vc{X}},\vc{n}) = G(\widetilde{\vc{X}},\proj_{n_0}(\vc{n})) =
\begin{cases}
G(\widetilde{\vc{X}},\vc{n}),  \text{ if } \|\vc{n}\|_{\infty} \le n_0 \\
G(\widetilde{\vc{X}},\proj_{n_0}(\vc{n})),  \text{ otherwise.}
\end{cases}
\end{equation*}
It is easy to verify that $\forall\ \vc{n}, \vc{m}\in\naturals^{|\sources|}\ \text{and}\ \vc{n}\leq\vc{m}$, we have $G_{n_0}(\widetilde{\vc{X}},\vc{n} + k\vc{e}_s) - G_{n_0}(\widetilde{\vc{X}},\vc{n})\geq G_{n_0}(\widetilde{\vc{X}}, \vc{m} + k\vc{e}_s) - G_{n_0}(\widetilde{\vc{X}},\vc{m})
$. Thus, $G_{n_0}(\widetilde{\vc{X}},\vc{n})$ is DR-submodular w.r.t. $\vc{n}$.

Consider $g_{n_0}(\vc{n}) = \mathbb{E}_{\widetilde{\vc{X}}} [G_{n_0}(\widetilde{\vc{X}},\vc{n})|\vc{n}]$. As an expectation of 'finite $\widetilde{\vc{X}}$', $g_{n_0}(\vc{n})$ is still DR-submodular, i.e., $\forall\ \vc{n}, \vc{m}\in\naturals^{|\sources|}, \ \vc{n}\leq\vc{m}\ \text{and}\ \forall\ n_0 \in \naturals$, we have $g_{n_0}(\vc{n} + k\vc{e}_s) - g_{n_0}(\vc{n})\geq g_{n_0}(\vc{m} + k\vc{e}_s) - g_{n_0}(\vc{m})$. 
Observe that $\lim_{n_0 \to \infty} g_{n_0}(\vc{n}) = g(\vc{n})$, $\forall\ \vc{n}'\in\naturals^{|\sources|}$. In fact, $\forall\ \vc{n}'\in\naturals^{|\sources|}$, $\exists \ n_0(\vc{n}') = \|\vc{n}'\|_{\infty}$, s.t. $g_{n_0}(\vc{n}') = g(\vc{n}')$. Furthermore, $\exists n_0 = \max\{ \|\vc{n}\|_{\infty}, \|\vc{n} + k\vc{e}_s\|_{\infty}, \|\vc{m}\|_{\infty}, \|\vc{m} + k\vc{e}_s\|_{\infty}\}$, s.t. $g_{n_0}(\vc{n}) = g(\vc{n})$, $g_{n_0}(\vc{n} + k\vc{e}_s) = g(\vc{n} + k\vc{e}_s)$, $g_{n_0}(\vc{m}) = g(\vc{m})$, $g_{n_0}(\vc{m} + k\vc{e}_s) = g(\vc{m} + k\vc{e}_s)$, and $g(\vc{n} + k\vc{e}_s) - g(\vc{n}) \geq g(\vc{m} + k\vc{e}_s) - g(\vc{m})$.

Similarly, we could obtain the monotonicity.
\end{proof}
Then, we establish the positivity of the gradient and non-positivity of the Hessian of $U$, following Thm.~1 in \cite{liu2022experimental} and utilizing Lem.~\ref{lem: submodular} to prove our Thm.~\ref{thm: submodular}.
\begin{proof}
By the law of total expectation:
    \begin{equation*}
        U^{\learner} (\vc{\lambda}^{\learner}) = \mathbb{E}[g^{\learner}(\vc{n}^{\learner})] = \sum_{t=0}^\infty \mathbb{E}\left[g(\vc{n})|n_{\source}^{\learner}= t\right]\cdot \frac{(\lambda_{\source}^{\learner} T)^{t} e^{-\lambda_{\source}^{\learner} T}}{t!}.
    \end{equation*}
Thus the first partial derivatives are:
\begin{align*}
     \frac{\partial U}{\partial \lambda_{\source, \type}^{\p}} = &\frac{\partial U^{\learner}}{\partial \lambda_{\source, \type}^{\p}} 
    =\sum_{t=0}^\infty \mathbb{E}\left[g(\vc{n})|n_{\source}^{\learner}= t\right]\cdot (\frac{t}{\lambda_{\source}^{\learner}}-T)\frac{(\lambda_{\source}^{\learner} T)^{t} e^{-\lambda_{\source}^{\learner} T}}{t!} \\
    = & \sum_{t=0}^\infty(\mathbb{E}\left[g(\vc{n})|n_{\source}^{\learner}= t+1\right] - \mathbb{E}\left[g(\vc{n})|n_{\source}^{\learner}= t\right])\cdot 
    \frac{(\lambda_{\source}^{\learner})^t T^{t+1}}{t!}e^{-\lambda_{\source}^{\learner} T}\ge 0
\end{align*}
by monotonicity of $g$. Note that, here $\type = \type^{\learner}$ and $\learner = \p[-1]$. Next, we compute the second partial derivatives $\frac{\partial^2 U}{\partial \lambda_{\source, \type}^{\p} \partial \lambda_{\source', \type'}^{\p'}}$. It is easy to see that for $\learner \neq \learner'$, we have:
\begin{equation*}
    \frac{\partial^2 U}{\partial \lambda_{\source, \type}^{\p} \partial \lambda_{\source', \type'}^{\p'}} = 0.
\end{equation*}
For $\learner = \learner'$, which also implies $\type = \type'$, and $\source = \source'$. No matter $\p$ equals to $\p'$ or not, it holds that
\begin{align*}
    & \frac{\partial^2 U}{\partial \lambda_{\source, \type}^{\p} \partial \lambda_{\source, \type}^{\p'}} = \frac{\partial ^2 U^{\learner}}{\partial (\lambda_{\source, \type}^{\p})^2} \\
    = & (\mathbb{E}\left[g(\vc{n})|n_{\source}^{\learner}= 1\right] - \mathbb{E}\left[g(\vc{n})|n_{\source}^{\learner}= 0\right])\cdot (-T^2)e^{-\lambda_{\source}^{\learner} T} +  
    \sum_{t=1}^\infty(\mathbb{E}\left[g(\vc{n})|n_{\source}^{\learner}= t+1\right] - \mathbb{E}\left[g(\vc{n})|n_{\source}^{\learner}= t\right])\cdot\ldots \\&\quad(\frac{(\lambda_{\source}^{\learner})^{t-1}T^{t+1}}{(t-1)!} - \frac{(\lambda_{\source}^{\learner})^t T^{t+2}}{t!}) e^{-\lambda_{\source}^{\learner} T} \\ 
    = & \sum_{t=1}^\infty\left((\mathbb{E}\left[g(\vc{n})|n_{\source}^{\learner}= t+1\right] - \mathbb{E}\left[g(\vc{n})|n_{\source}^{\learner}= t\right]) - (\mathbb{E}\left[g(\vc{n})| \right. \right. \left. \left. n_{\source}^{\learner}= t\right] - \mathbb{E}\left[g(\vc{n})| n_{\source}^{\learner}= t-1\right])\right)\cdot \frac{(\lambda_{\source}^{\learner})^{t-1} T^{t+1}}{(t-1)!} e^{-\lambda_{\source}^{\learner} T}\\
    \le& 0,
\end{align*}
by the submodularity of $G$. For $\source\neq \source'$, which also implies $\p \neq \p'$
\begin{align*}
     \frac{\partial^2 U}{\partial \lambda_{\source, \type}^{\p} \partial \lambda_{\source', \type}^{\p'}} =&
    \frac{\partial^2 U^{\learner}}{\partial \lambda_{\source, \type}^{\p} \partial \lambda_{\source', \type}^{\p'}} 
    = \sum_{t=0}^\infty\sum_{k=0}^\infty \mathbb{E}\left[g(\vc{n})|n_{\source}^{\learner}= t,  n_{\source'}^{\learner} = k\right]\cdot\frac{\partial}{\partial\lambda_{\source}^{\learner}}\mathrm{P}(n_{\source}^{\learner} = t)\cdot 
    \frac{\partial}{\partial\lambda_{\source'}^{\learner}}\mathrm{P}(n_{\source'}^{\learner} = k) \\
    = & \sum_{t=0}^\infty\sum_{k=0}^\infty ((\mathbb{E}\left[g(\vc{n})|n_{\source}^{\learner}= t+1, n_{\source'}^{\learner} = k+1\right] - 
    \mathbb{E}\left[g(\vc{n})|n_{\source}^{\learner}= t, n_{\source'}^{\learner} = k+1\right])  - \ldots\\
    &\quad (\mathbb{E}\left[g(\vc{n})|n_{\source}^{\learner}= t+1, \right.
    \left. n_{\source'}^{\learner} = k\right] -\mathbb{E}\left[g(\vc{n})|n_{\source}^{\learner}= t, n_{\source'}^{\learner} = k\right]))\cdot
    \frac{(\lambda_{\source}^{\learner})^t (\lambda_{\source'}^{\learner})^k T^{t+k+2}}{t!k!}  e^{-(\lambda_{\source}^{\learner}+\lambda_{\source'}^{\learner}) T}\\\le& 0
\end{align*}
also by the submodularity of $g$.
\end{proof}

\section{Proof of Lemma \ref{lem: HEAD bound4}}
\label{app: proof_HEAD bound4}

We know that Poisson has a subexponential tail bound~\cite{ccanonne2017note}, then, after truncating, the $\mathrm{HEAD}$, defined in Eq. \eqref{eq: HEAD} is guaranteed to be within a constant factor from the true partial derivative:
\begin{lemma}[Lemma 4 in \cite{liu2022experimental}]
\label{lem: HEAD bound}
For $h(u) = 2\frac{(1+u)\ln{(1+u)}-u}{u^2}$ and $n'\geq \lambda_{\source}^{\learner} T$, we have:
%\begin{align}
   $$ \frac{\partial U}{\partial\lambda_{\source, \type}^{\p}} \ge \mathrm{HEAD}_{\source, \type}^{\p}(n')
    \ge 
    (1 - \prob[n^\learner_{\source}\geq n'+1]) \frac{\partial U}{\partial\lambda_{\source, \type}^{\p}}.$$
%\end{align}
\end{lemma}

To quantify the distance between the estimated gradient and $\mathrm{HEAD}$, we first introduce two auxiliary lemmas as follows.
By Chernoff bounds described by Thm~A.1.16 in \cite{alon2004probabilistic}, we get:
\begin{lemma}
\label{lem: HEAD bound2}
If there exists a constant vector $\vc{c} = [c_{\source}]_{\source \in \sources} \in \reals^{|\sources|}$, such that for any $\source\in\sources$ and $i \le n_{\source}^{\learner}$, $\|\x^{\learner}_{\source,i}\|_2 \le c_{\source}$, then:
\begin{equation*}
 \left|\widehat{\frac{\partial U}{\partial\lambda_{\source, \type}^{\p}} } - \mathrm{HEAD}_{\source, \type}^{\p} \right| \le \gamma \max_{\learner \in \learners, \source \in \sources} \log(1 + \frac{\lambda_{\mathrm{MAX}}(\vc{\Sigma}_0^\learner) c_{\source}^2}{\sigma_{\source,\type}^2} ),
\end{equation*}
with probability greater than $1-2\cdot e^{-\gamma^2 N_1 N_2/2T^2(n'+1)}$, where $\lambda_{\mathrm{MAX}}(\vc{\Sigma}_0^\learner)$ is the maximum eigenvalue of matrix $\vc{\Sigma}_0^\learner$.
\end{lemma}
\begin{proof}
We define 
\begin{align*}
    X^{j,k}(n) = & \frac{G^{\learner}(\vc{X}^{\learner,j,k},\vc{n}^{\learner,j}|_{n_{\source}^{\learner,j} = n+1}) - G^{\learner}(\vc{X}^{\learner,j,k},\vc{n}^{\learner,j}|_{n_{\source}^{\learner,j} = n})  - \Delta_{\source}^{\learner}(\vc{\lambda}^\learner,n)}{G_\mathrm{MAX}},
\end{align*}
where $G_{\mathrm{MAX}} = \max_{\learner \in \learners, \source\in \sources, \x\in B(\vc{0};c)}(G^{\learner}(\x, \vc{e}_\source)-G^{\learner}(\x, \vc{0})) = \max_{\learner \in \learners,\source \in \sources , \x\in B(\vc{0};c)} \log \big( 1 + \frac{1}{\sigma_{\source,\type}^2} \x_{\source}^{\top} \vc{\Sigma_\learner} \x_{\source} \big)$. We have $|X^{j,k}(n)| \le 1$, because:
\begin{align*}
    G_{\mathrm{MAX}} & \ge (G^{\learner}(\x, \vc{e}_\source)-G^{\learner}(\x, \vc{0})) \ge G(\vc{X}^{\learner,j,k},\vc{n}^{j}|_{n_{\source}^{\learner,j} = n+1}) - G(\vc{X}^{\learner,j,k},\vc{n}^{j}|_{n_{\source}^{\learner,j} = n}),
\end{align*}
for any $\learner\in\learners$, $\source\in \sources$, $\x\in B(\vc{0};c)$, $n\geq 0$. By Chernoff bounds described by Theorem A.1.16 in \cite{alon2004probabilistic}, we have
$$\prob\left[\left| \sum_{j=1}^{N_1} \sum_{k=1}^{N_2} \sum_{n = 0}^{n = n^{\prime}} X^{j,k}(n) \right|> c\right] \leq 2e^{-c^2/2N_1N_2(n^{'}+1)}.$$ 
Suppose we let $c = \gamma \cdot N_1N_2/T$, where $\gamma$ is the step size, then we have 
\begin{align*}
\left|\widehat{\frac{\partial U}{\partial\lambda_{\source, \type}^{ \p}} } - \mathrm{HEAD}_{\source, \type}^{ \p} \right|  \leq & \left|\sum_{n=0}^{n^{\prime}} \sum_{j=1}^{N_1} \sum_{k=1}^{N_2} \frac{G^{\learner}(\vc{X}^{\learner,j,k},\vc{n}^{\learner,j}|_{n_{\source}^{\learner,j} = n+1}) - G^{\learner}(\vc{X}^{\learner,j,k},\vc{n}^{\learner,j}|_{n_{\source}^{\learner,j} = n}) - \Delta_{\source}^{\learner}(\vc{\lambda}^\learner,n)}{N_1 N_2} T\right| \\
= & \left|\sum_{t=0}^{n^{\prime}} \sum_{j=1}^{N_1} \sum_{k=1}^{N_2} X^{j,k}(n)\right|\cdot \frac{T}{N_1 N_2}\cdot G_\mathrm{MAX}\notag \leq \gamma \cdot G_\mathrm{MAX} \\
= & \gamma \max_{\learner \in \learners,\source \in \sources , \x\in B(\vc{0};c)} \log \big( 1 + \frac{1}{\sigma_{\source,\type}^2} \x_{\source}^{\top} \vc{\Sigma_\learner} \x_{\source} \big) \\
\le & \gamma \max_{\learner \in \learners, \source \in \sources , \x\in B(\vc{0};c)} \log(1 + \frac{1}{\sigma_{\source,\type}^2} \lambda_{\mathrm{MAX}}(\vc{\Sigma}_\learner)\|\x_{\source}\|_2^2) \\
= & \gamma \max_{\learner \in \learners, \source \in \sources} \log(1 + \frac{\lambda_{\mathrm{MAX}}(\vc{\Sigma}_\learner)c_{\source}^2}{\sigma_{\source,\type}^2} ),
\end{align*}
with probability greater than $1-2\cdot e^{-\gamma^2 N_1 N_2/2T^2(n'+1)}$, and the last inequality holds because $\lambda_{\mathrm{MIN}}(\Sigma_{\learner})|\x|^2 \le \x^{\top} \Sigma_{\learner} \x \le \lambda_{\mathrm{MAX}}(\Sigma_{\learner})|\x|^2$, for any $\x$, by \cite{horn2012matrix}.
\end{proof}

By sub-Gaussian norm bound \cite{Lectures36709S19}, we get:
\begin{lemma}
[Theorem 8.3 in \cite{Lectures36709S19}]
\label{lem: sub_Gaussian bound}
If $\x \sim N(0,\vc{\Sigma})$, for any $\delta \in (0,1)$:
\begin{equation*}
\begin{split}
 \prob\left[\|\x\|_2\le 4 \sqrt{\lambda_{\mathrm{MAX}}(\vc{\Sigma})} \sqrt{d} + 2 \sqrt{\lambda_{\mathrm{MAX}}(\vc{\Sigma})} \sqrt{\log \frac{1}{\delta}}\right] \ge 1 - \delta,
\end{split}
\end{equation*}
where $d$ is the dimension of $\x$.
\end{lemma}
Then, we bound the distance between the estimated gradient and $\mathrm{HEAD}$, and prove it by law of total probability theorem.
\begin{lemma}
\label{lem: HEAD bound3}
For any $\delta \in (0,1)$, and $n'\geq \lambda_{\source}^{\learner} T$,
\begin{equation*}
    \big|\widehat{\frac{\partial U}{\partial\lambda_{\source, \type}^{\p}} } - \mathrm{HEAD}_{\source, \type}^{\p} \big| \le \gamma \max_{\learner \in \learners, \source \in \sources} \log(1 +  \frac{\lambda_{\mathrm{MAX}}(\vc{\Sigma}_0^\learner)c_\source^2}{\sigma_{\source,\type}^2}),
\end{equation*}
where  $c_\source = 4 \sqrt{\lambda_{\mathrm{MAX}}(\vc{\Sigma}_\source)} \sqrt{d} + 2 \sqrt{\lambda_{\mathrm{MAX}}(\vc{\Sigma}_\source)} \sqrt{\log \frac{1}{\delta}}$, $\gamma$ is stepsize, and $d$ is the dimension of feature $\x$, with probability greater than $1-2 \cdot e^{-\gamma^2 N_1 N_2/2T^2(n'+1)}- |\sources|n'\delta - (|\sources|-1)\delta_{\source, \type}^{\p}$, and $\delta_{\source, \type}^{\p} = \max_{\source'\in \sources\setminus\{\source\}} \prob[n^\learner_{\source'}\geq n'+1] $.
\end{lemma}
\begin{proof}
According to law of total probability theorem, where event $A = \left|\widehat{\frac{\partial U}{\partial\lambda_{\source, \type}^{ \p}} } - \mathrm{HEAD}_{\source, \type}^{ \p} \right| > \gamma \max_{\learner \in \learners, \source \in \sources} \log(1 + \\ \frac{\lambda_{\mathrm{MAX}}(\vc{\Sigma}_\learner)c_{\source}^2}{\sigma_{\source,\type}^2})$, event $B=\{\|\x^{\learner}_{\source,i}\|_2 \le c_{\source}: \forall \source\in\sources,i \le n_{\source}^{\learner}\}$, and event $C=\{n_{\source}^{\learner} \le n': \forall \source\in\sources\}$, the probability of event A is bounded by:
\begin{align*}
    \prob(A) = & \prob(A|B \cap C)\prob(B \cap C) + \prob(A|\overline{B \cap C})\prob(\overline{B \cap C}) 
    \le \prob(A|B \cap C)\prob(B \cap C) + \prob(\overline{B \cap C}).
\end{align*}
Then, the probability of its complement is bounded by:
\begin{align*}
   \prob(\overline{A}) \ge & 1 - \left(\prob(A|B \cap C)\prob(B \cap C) + \prob(\overline{B \cap C})\right)
   = \prob(B\cap C) - \prob(A|B \cap C)\prob(B \cap C) \\
   = & ~\prob(\overline{A}|B \cap C)\prob(B\cap C) 
   \ge (1-2\cdot e^{-\gamma^2 N_1 N_2/2T^2(n'+1)}) \cdot (1-\delta)^{|\sources|n'} \cdot  \prod_{\source'\in \sources\setminus\{\source\}} (1-\prob[n^\learner_{\source'}\geq n'+1]),
\end{align*}
where the first term comes from Lem.~\ref{lem: HEAD bound2}, the second term is by Lem.~\ref{lem: sub_Gaussian bound}, and the last term derives from the condition, where the number of arrivals $n_{\source}^{\learner}$ from source $\source$ is less then or equal to $n'$. Since truncating has forced one coordinate (coordinate $s$) to be less then or equal to $n'$ already, the condition is applied to the left sources.

The probability can be further simplified by Bernoulli's inequality:
\begin{align*}
    \prob(\overline{A}) \ge & (1-2\cdot e^{-\gamma^2 N_1 N_2/2T^2(n'+1)}) \cdot (1-|\sources|n'\delta) \cdot  (1-(|\sources|-1)\delta_{\source, \type}^{ \p}) \\
    \ge & 1 - 2\cdot e^{-\gamma^2 N_1 N_2/2T^2(n'+1)} - |\sources|n'\delta - (|\sources|-1)\delta_{\source, \type}^{ \p},
\end{align*}
where $\delta_{\source, \type}^{ \p} = \max_{\source'\in \sources\setminus\{\source\}}\{\prob[n^\learner_{\source'}\geq n'+1]\}$.
\end{proof}

By Lems.~\ref{lem: HEAD bound} and \ref{lem: HEAD bound3}, we instantly get the distance between the estimated and true gradient, shown in Lem. \ref{lem: HEAD bound4}

\section{Proof of Theorem \ref{thm: FW}}
\label{app: proof_FW}

We bound the quality of our gradient estimator by Lem.~\ref{lem: HEAD bound4}:
\begin{lemma}
\label{lem: estimation guarantee} 
At each iteration $k$,  with probability greater than $1-(2\cdot e^{-\gamma^2 N_1 N_2/2T^2(n'+1)} + |\sources|n'\delta + (|\sources|-1)\mathrm{P}_\mathrm{MAX})\cdot P_\SR$,
\begin{equation}
\label{eq: estimator guarantee full}
    \langle \vc{v}(k),\nabla U(\mathcal{\vc{\lambda}}(k))\rangle \geq a\cdot\max_{\vc{v}\in\mathcal{D}}\langle \vc{v},\nabla U(\mathcal{\vc{\lambda}}(k))\rangle - b,
\end{equation}
where 
\begin{align}
    &a = 1- \mathrm{P}_\mathrm{MAX}, \quad \text{and}\\
    &b = 2\lambda_{\mathrm{MAX}} \cdot \gamma \max_{\learner \in \learners, \source \in \sources} \log(1 + \frac{\lambda_{\mathrm{MAX}}(\vc{\Sigma}_0^\learner)c_{\source}^2}{\sigma_{\source,\type}^2} ),
\end{align}
for $\mathrm{P}_\mathrm{MAX}=\max_{k=1,\dots, K} \mathrm{P}_\mathrm{MAX}(k) = \max_{l\in\learners,\source\in\sources} \mathrm{P}[n^{\learner}(k)_{\source}\geq n^{\prime} + 1]$ ($n^{\learner}_{\source}(k)$ is a Poisson r.v. with parameter $\lambda^{\learner}_{\source}(k) T$), and $\lambda_\mathrm{MAX} = \max_{\vc{\lambda}\in\mathcal{D}}\|\vc{\lambda}\|_1$.
\end{lemma}

\begin{proof}
We now use the superscript $k$ to represent the parameters for the $k$th iteration: we find $\vc{v}^k\in\mathcal{D}$ that maximizes $\langle \vc{v}^k,\widehat{\nabla U(\vc{\lambda}^k)}\rangle$. 
Let $\vc{u}^{k}\in\mathcal{D}$ be the vector that maximizes $\langle \vc{u}^k,\nabla U(\vc{\lambda}^k)\rangle$ instead.
We have
\begin{align*}
    \langle \vc{v}^k,\nabla U(\vc{\lambda}^k)\rangle 
    \geq & \langle \vc{v}^k,\widehat{\nabla U(\vc{\lambda}^k)}\rangle - \lambda_{\text{MAX}} \cdot \gamma \max_{\learner \in \learners, \source \in \sources} \log(1 + \frac{\lambda_{\mathrm{MAX}}(\vc{\Sigma}_\learner)c_{\source}^2}{\sigma_{\source,\type}^2} )\\
    \geq & \langle \vc{u}^{k},\widehat{\nabla U(\vc{\lambda}^k)}\rangle - \lambda_{\text{MAX}} \cdot \gamma \max_{\learner \in \learners, \source \in \sources} \log(1 + \frac{\lambda_{\mathrm{MAX}}(\vc{\Sigma}_\learner)c_{\source}^2}{\sigma_{\source,\type}^2} )\\
    \geq & (1-\mathrm{P_{MAX}})\cdot\langle \vc{u}^{k},\nabla U(\vc{\lambda}^k)\rangle - 2\lambda_{\text{MAX}} \cdot \gamma \max_{\learner \in \learners, \source \in \sources} \log(1 + \frac{\lambda_{\mathrm{MAX}}(\vc{\Sigma}_\learner)c_{\source}^2}{\sigma_{\source,\type}^2} ),
\end{align*}
where the first and last inequalities are due to Lem.~\ref{lem: HEAD bound4} and the second inequality is because $\vc{v}^k$ maximizes $\langle \vc{v}^k,\widehat{\nabla U(\vc{\lambda}^k)}\rangle$. The above inequality requires the satisfaction of Lem.~\ref{lem: HEAD bound4} for every partial derivative. By union bound, the above inequality satisfies with probability greater than $1- (2\cdot e^{-\gamma^2 N_1 N_2/2T^2(n'+1)} + |\sources|n'\delta + \sum_{\source\in\sources,\type \in\types}\sum_{ \p \in \paths_{\source,\type}}(|\sources|-1)\delta_{\source, \type}^{ \p}) \cdot P_\SR$, and thus greater than $1-(2\cdot e^{-\gamma^2 N_1 N_2/2T^2(n'+1)} + |\sources|n'\delta + (|\sources|-1)\mathrm{P}_\mathrm{MAX})\cdot P_\SR$.
\end{proof}

By Lemma 2 in \cite{liu2022experimental}, and Lipschitz constant of $\nabla U$: $L = 2 T^2 P_\SR \cdot \allowbreak \max_{\learner \in \learners, \source \in \sources} \log(1 + \frac{\lambda_{\mathrm{MAX}}(\vc{\Sigma}_0^\learner)c_{\source}^2}{\sigma_{\source,\type}^2} )$, we get:
\begin{lemma}
With probability greater than $1- (2\cdot e^{-\gamma^2 N_1 N_2/2T^2(n'+1)} \\ - |\sources|n'\delta - (|\sources|-1)\mathrm{P}_\mathrm{MAX})\cdot P_\SR K$, the output solution $\vc{\lambda}(K)\in \mathcal{D}$ satisfies:
\begin{equation*}
\label{eq: final result short}
\begin{split}
    &U(\vc{\lambda}(K)) \ge (1 - e^{\mathrm{P}_\mathrm{MAX} -1 } ) \max_{\vc{\lambda}\in\mathcal{D}}U(\vc{\lambda}) -   (T^2 P_\SR \lambda_{\mathrm{MAX}}^2 + 2 \lambda_{\mathrm{MAX}}) \gamma \max_{\learner \in \learners, \source \in \sources} \log(1 + \frac{\lambda_{\mathrm{MAX}}(\vc{\Sigma}_0^\learner)c_{\source}^2}{\sigma_{\source,\type}^2} ) ,
\end{split}
\end{equation*}
where $K = \frac{1}{\gamma}$ is the number of iterations.
\end{lemma}

Then, we utilize subexponential Poisson tail bound described by Thm.~1 in \cite{ccanonne2017note}, and organize constants in above lemma to obtain Thm.~\ref{thm: FW}.
\begin{proof}
    From Thm.~1 in \cite{ccanonne2017note}, the probability is an increasing function w.r.t. $\lambda_{\source}^\learner$, and $\lambda_\mathrm{MAX}$ is an upper bound for $\lambda_{\source}^\learner$. Letting $u  = \frac{n' - \lambda_\mathrm{MAX}T}{\lambda_\mathrm{MAX}T}$, we have:
\begin{align*}
    \mathrm{P_{MAX}} & < e^{-\frac{(n' - \lambda_\mathrm{MAX}T)^2}{2\lambda_\mathrm{MAX}T} h(\frac{n' - \lambda_\mathrm{MAX}T}{\lambda_\mathrm{MAX}T})}  = e^{-\lambda_\mathrm{MAX}T( (1+u)\ln(1+u) -u )} \le \Omega(e^{-\lambda_\mathrm{MAX}T u}) = \epsilon_1,
\end{align*}
where the last inequality holds because $(1+u) \ln(1+ u) - u > u$ when $u$ is large enough, e.g., $u \ge 4$. Thus, $n' = O( \lambda_\mathrm{MAX} T + \ln \frac{1}{\epsilon_1})$.
Then, we have 
\begin{align*}
    (2\cdot e^{-\gamma^2 N_1 N_2/2T^2(n'+1)}- |\sources|n'\delta  - (|\sources|-1)\mathrm{P}_\mathrm{MAX})\cdot 
    P_\SR K = \epsilon_0,
\end{align*}
and assume that 
\begin{align*}
    \begin{cases}
    2\cdot e^{-\gamma^2 N_1 N_2/2T^2(n'+1)} \le \frac{\epsilon_0}{3 P_\SR K}  \\
    |\sources|n'\delta \le \frac{\epsilon_0}{3 P_\SR K} \\
    (|\sources|-1)\mathrm{P}_\mathrm{MAX} \le \frac{\epsilon_0}{3 P_\SR K}
    \end{cases}
\end{align*}
Then, we first get: $K = O(\frac{\epsilon_0}{P_\SR  (|\sources|-1)} \epsilon_1)$, then $\delta = O(\frac{\epsilon_0}{P_\SR K |\sources|n'})$, and $N_1 N_2 = \Omega(\ln \frac{P_\SR K}{\epsilon_0} \cdot T^2(n'+1)K^2)$.
Finally, we have
\begin{align*}
    \epsilon_2 = & (T^2 P_\SR \lambda_{\mathrm{MAX}}^2 + 2 \lambda_{\mathrm{MAX}}) \frac{1}{K} \cdot  \max_{\learner \in \learners, \source \in \sources} \log(1 + \frac{\lambda_{\mathrm{MAX}}(\vc{\Sigma}_\learner)c_{\source}^2}{\sigma_{\source,\type}^2} ),
\end{align*}
where $c_\source = 4 \sqrt{\lambda_{\mathrm{MAX}}(\vc{\Sigma}_\source)} \sqrt{d} + 2 \sqrt{\lambda_{\mathrm{MAX}}(\vc{\Sigma}_\source)} \sqrt{\log \frac{1}{\delta}}$.
\end{proof}
\fussy}{}

\end{document}